%% file: main.tex
\documentclass[11pt]{article}

\input{usepackages}
\input{commands}
\input{mytheorems}
\input{comments}

\title{Distributed Quantum Property Testing with Communication Constraints}

\author[1]{Mina Doosti}
\author[2,3,4]{Ryan Sweke}
\author[1]{Chirag Wadhwa}
\affil[1]{School of Informatics, University of Edinburgh}
\affil[2]{African Institute for Mathematical Sciences (AIMS), South Africa}
\affil[3]{Department of Mathematical Sciences, Stellenbosch University, Stellenbosch 7600, South Africa}
\affil[4]{National Institute for Theoretical and Computational Sciences (NITheCS), South Africa}
\date{}

\begin{document}
\maketitle

\begin{abstract}
We introduce a framework for distributed quantum inference under communication constraints. In our model, $m$ distributed nodes each receive one copy of an unknown $d$-dimensional quantum state $\rho$, before communicating via a constrained one-way communication channel with a central node, which aims to infer some property of $\rho$. This framework generalizes the classical distributed inference framework introduced by Acharya, Canonne, and Tyagi [COLT2019], by allowing quantum resources such as quantum communication and shared entanglement. Within this setting, we focus on the fundamental problem of quantum state certification: Given a complete description of some state $\sigma$, decide whether $\rho=\sigma$ or $\|\rho-\sigma\|_1\geq \epsilon$. Additionally, we focus on the case of limited \textit{quantum} communication between distributed nodes and the central node. We show that when each communication channel is limited to only $n_q\leq \log d$ qubits, then the sample complexity of distributed state certification is $\mathcal{O}(\nicefrac{d^2}{2^{n_q}\epsilon^2})$ when public randomness is available to all nodes. Moreover, under the assumption that the channels used by the distributed nodes are \emph{mixedness-preserving}, we prove a matching lower bound. We further demonstrate that shared randomness is necessary to achieve the above complexity, by proving an $\Omega(\nicefrac{d^3}{4^{n_q} \epsilon^2})$ lower bound in the \emph{private-coin} setting under the same assumption as above. Our lower bounds leverage a recently introduced quantum analogue of the celebrated Ingster--Suslina method and generalize arguments from the classical setting. Together, our work provides the first characterization of distributed quantum state certification in the regime of limited quantum communication and establishes a general framework for distributed quantum inference with communication constraints. 
\end{abstract}
\newpage

\tableofcontents
\newpage

\section{Introduction}\label{s:introduction}

Recent years have seen a massive surge in the study of algorithms for testing and learning quantum states~\cite{gs007,anshu2023surveycomplexitylearningquantum}.  Among other things, this development has been spurred by the fact that such algorithms are necessary for both the characterization of emerging quantum computational devices and foundational scientific applications of quantum computing. This line of work has resulted in tight characterizations of the resource requirements for a wide variety of both testing and learning problems, under many different assumptions on the algorithm, such as its ability to perform adaptive measurements and its ability to perform coherent measurements on multiple copies of the unknown quantum state.  However, the vast majority of previous work in this area has been in the \textit{centralized} setting, where the algorithm itself has access, in some way, to copies of the unknown quantum state. 

In this work, we study the problem of learning and testing quantum states in the \textit{distributed} setting, where the copies of the unknown quantum state(s) are distributed among multiple distributed nodes, all of which communicate with a central node running a learning/testing algorithm. This setting has a variety of motivations. Firstly, as quantum communication networks develop, this setting is natural and allows us to characterize the potential and limitations of inference over such networks. Additionally, this problem is the natural quantum analogue of classical statistical inference in the distributed setting, whose theoretical foundations have proven essential for the development of large-scale federated learning protocols, in which data is distributed among multiple nodes or data centers.

As per the classical setting, when trying to define a concrete framework for the analysis of \textit{distributed} quantum property testing, one is immediately faced with a variety of choices, such as the following:
\begin{enumerate}
\item What are the constraints on the communication channels between the distributed and central nodes? Are the communication channels classical or quantum? How many bits or qubits can be sent down each channel? Are the channels required to ensure some notion of privacy?
\item What shared resources are available to the distributed nodes? Do they have access to public randomness, or shared entangled states?
\item What communication is allowed between the distributed nodes?
\end{enumerate}
Most of these choices are not merely of theoretical interest; they capture real limitations in networks, especially when quantum resources are involved. Among them, communication constraints are particularly fundamental, since communication is often one of the main bottlenecks in quantum networks and distributed quantum protocols such as quantum secure multiparty computation and delegated quantum computing.

A variety of concrete models, each defined by different answers to the above questions, have been illustrated in Figure~\ref{fig:framework}. However, we can immediately make the following observations. Firstly, when \textit{unlimited quantum communication} is allowed between the distributed and central nodes, one immediately recovers the unconstrained centralized setting. Indeed, in this case, each distributed node can simply send their quantum state to the central node, which can then run any testing algorithm -- even one requiring adaptive or multi-copy coherent measurements. Secondly, in the case where the communication channels between distributed nodes and the central node are purely classical, but \textit{unlimited classical communication} is involved, then one recovers the centralized setting in which only \textit{single-copy} measurements are possible. More specifically, in this case, each distributed node can do the single-copy measurements required by the centralized testing algorithm, then send the classical measurement results to the central node for classical post-processing. Taken together, we have the following observation:
\begin{center}
\textit{The problem of distributed quantum inference is primarily interesting under communication constraints!}
\end{center} 
In light of this, our attention here is focused precisely on inference in this setting of limited communication. Moreover, for concreteness, we focus primarily on the fundamental testing problem of \textit{quantum state certification} (see \Cref{def:quantum-state-certification}), the quantum analogue of \textit{distribution identity testing}~\cite{gs009}. The copy complexity of this problem is well-understood in the centralized setting \cite{o2015quantum,buadescu2019quantum}. With this in mind, we are concerned with the following concrete question:

\begin{center}
\textit{How do communication constraints impact the copy complexity of distributed state certification?}
\end{center}

\begin{figure}
    \centering
     \includegraphics{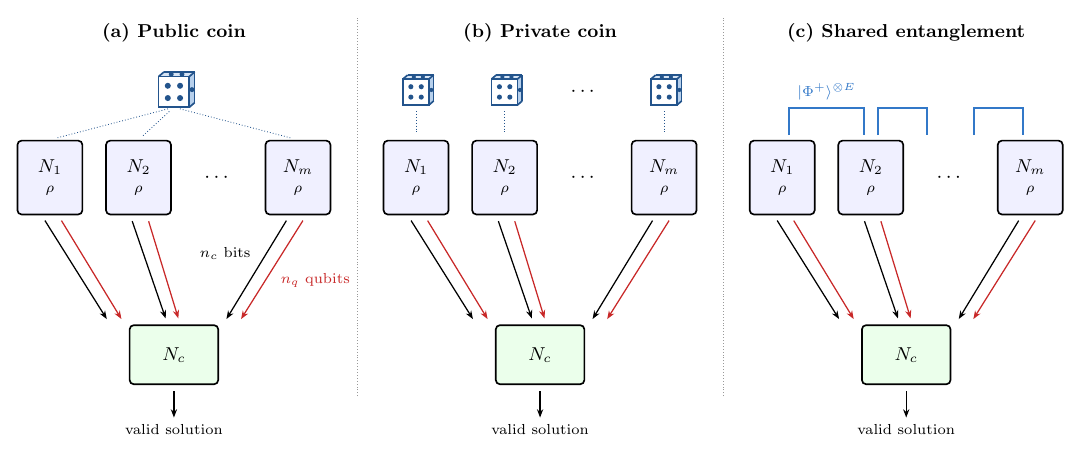} 
    \caption{An illustration of the $(n_c,n_q,R,E)$ model for distributed quantum inference, as per Definition~\ref{def:distributed-model}. Each distributed node $\{N_i\}_{i\in [m]}$ holds a single copy of $\rho$ and communicates with the central node $N_c$ via a communication channel limited to $n_c$ bits and $n_q$ qubits. The central node should output a valid solution to the inference problem (eg ``Accept'' or ``Reject'' in the case of property testing, or a valid hypothesis in the case of learning). We use $R\in \{\mathsf{public},\mathsf{private}\}$ to distinguish between the public and private coin setting (illustrated in panels (a) and (b) respectively), and $E$ to denote the number of Bell pairs shared between each neighbouring distributed node.} 
    \label{fig:framework} 
\end{figure}

\subsection{Framework}\label{ss:framework}

We consider quantum inference problems defined by access to multiple copies of some unknown quantum state $\rho$. As illustrated in Figure~\ref{fig:framework}, we consider the distributed setting in which $m$ distributed nodes $\{N_i\,|\,i\in[m]\}$ each hold a single copy of the unknown state $\rho$. Each such node communicates with a central node $N_c$, which should output a candidate solution to the problem. 

As discussed, there are a variety of different choices one can make regarding the communication channels between the distributed and central nodes, the communication channels between the distributed nodes, and the resources shared by the distributed nodes. In the present work, we do not allow any communication between distributed nodes $N_i$ and allow only one-way communication between distributed nodes and the central node. With this in mind, we define the $(n_c, n_q, R, E)$ model for distributed quantum inference as follows:

\begin{definition}[$(n_c, n_q, R, E)$-model for distributed quantum inference]\label{def:distributed-model} We define the $(n_c, n_q, R, E)$ model as the model in which no communication is allowed between distributed nodes, only one-way communication is allowed between distributed nodes and the central node, and:
\begin{enumerate}
\item At most $n_c$ classical bits and $n_q$ qubits can be sent from any distributed node $N_i$ to the central node $N_c$.
\item $R\in \{\mathsf{public},\mathsf{private}\}$ indicates whether all nodes only have access to a private source of randomness, or whether all nodes have access to a shared public source of randomness.
\item $E$ denotes the number of Bell pairs shared between each pair of neighboring nodes (assume the distributed nodes sit on the vertices of a graph). 
\end{enumerate}
\end{definition}

The models defined above should be viewed as a quantum generalization of the classical framework of distributed inference under information constraints introduced in~\cite{acharya2020inferenceinformationconstraints}. Their setting considers multiple players, each sending a single message to a central referee via a constrained communication channel, who then solves an inference task. In the absence of communication between players, and only one-way communication from players to referee, such as we consider here, this setting is known as the \emph{Simultaneous Message Passing (SMP)} model. Our formulation extends this paradigm from distributed classical samples to distributed quantum states, while also allowing for meaningful quantum resources, including quantum communication and shared entanglement through the parameters $(n_c, n_q, R, E)$. This includes the purely classical restriction where $n_q = 0$, $E = 0$, in which case our model recovers the public-coin/private-coin SMP framework studied extensively in the classical distributed inference literature (discussed in Section~\ref{ss:prior-work}).

We say that a quantum inference problem for an unknown state $\rho$ can be solved with $m$ distributed nodes in the $(n_c, n_q, R, E)$ distributed model if there exist algorithms $\mathcal{A}_i$ for each distributed node $N_i$ and an algorithm $\mathcal{A}_c$ for the central node $N_c$, all executable using the randomness and entanglement specified by the parameters $R,E$, such that:

\begin{enumerate}
\item On input $\rho$, algorithm $\mathcal{A}_i$ outputs a quantum-classical message $m_i = (\tilde{m}_i,\phi_i)$, where $\tilde{m}_i\in\{0,1\}^{n_c}$ and $\phi_i$ is an $n_q$-qubit quantum state,
\item and on input $\{m_i\}_i$, the central algorithm $\mathcal{A}_c$ outputs a valid solution to the problem with sufficiently high probability.
\end{enumerate}
As our distributed model assumes that each distributed node holds one copy of the unknown quantum state, we define the sample complexity of solving a quantum inference problem in the $(n_c, n_q, R, E)$ distributed model as the minimum number of nodes for which the problem can be solved. We note that formalizing the above definition for a specific inference problem requires only specifying what constitutes a ``valid solution''. For instance, ``Accept'' or ``Reject'' in the case of property testing, an $\epsilon$-accurate hypothesis in the case of learning, or a sufficiently accurate estimate in the case of property estimation.

While any quantum inference problem can be studied in this framework, we focus specifically on the problem of quantum state certification~\cite{buadescu2019quantum}, defined as follows:

\begin{definition}[Quantum State Certification]\label{def:quantum-state-certification} Given $\epsilon,\delta \in (0,1]$, a complete classical description of a quantum state $\sigma$, and multiple copies of an unknown quantum state $\rho$, an algorithm is said to succeed at $(\eps,\delta)$-certification of $\sigma$ if it behaves as follows:
\begin{enumerate} 
\item If $\rho=\sigma$, output ``Accept" with probability at least $1-\delta$.
\item If $\|\rho-\sigma\|_1\geq \epsilon $, output ``Reject" with probability at least $1-\delta$.
\end{enumerate}
When $\delta$ is unspecified, we mean the algorithm succeeds with probability at least $\frac{3}{4}$. 
\end{definition}

This is the natural quantum analogue of distribution identity testing, sometimes referred to as ``testing goodness-of-fit''~\cite{gs009}. We note that above we have used the \textit{trace distance} to determine the reject instances, but one can also consider other distance measures.

With this in mind, the specific question we seek to answer is the following:
\begin{center}
\textit{What is the sample complexity of solving quantum state certification in the $(n_c,n_q,R,E)$ distributed setting?}
\end{center}
However, we note that, for $n$-qubit input states $\rho$:
\begin{enumerate}
\item When $n_q\geq n$, the $(n_c,n_q,R,E)$ model is equivalent to the centralized setting in which the certification algorithm is allowed to make multi-copy coherent measurements. In particular, each distributed node can just send their state $\rho$ to the central node, which can then execute a multi-copy coherent measurement. In this setting, the complexity of state certification is fully characterized~\cite{o2015quantum,buadescu2019quantum,odonnell2025instanceoptimalquantumstatecertification}.
\item When $n_q = 0$ but $n_c\geq n$, the $(n_c,n_q,R,E)$ model is equivalent to the centralized setting in which the certification algorithm is allowed to make non-adaptive single-copy measurements. In particular, in this case each distributed node can just perform the desired single-copy measurement, and send the outcome to the central node for classical post-processing. In this setting, the complexity of state certification is again fully characterized~\cite{bubeck2020entanglement,chen2022toward,chen2022tightStateCertification}.
\item When $n_q = 0$ but $n_c< n$, the $(n_c,n_q,R,E)$ model is equivalent to the centralized setting in which the algorithm can only make non-adaptive single-copy measurements with a further restriction: each POVM can have at most $2^{n_c}$ outcomes. The complexity of state certification under this restriction has recently been characterized by Liu and Acharya~\cite{liu2024quantum}.
\end{enumerate}
As a result of the above observations, the interesting regime is the one in which $n_c < n$ and $0<n_q< n$. In this work, we focus primarily on the case of $n_c = 0$ and $0<n_q<n$, i.e., the setting where only limited quantum communication is possible. 

\subsection{Our results}\label{ss:results}
We summarize our results and those of relevant prior work in \Cref{tab:results}. Recall that our focus is on distributed state certification in the setting where only \emph{quantum} communication from the distributed nodes to the central node is permitted.  In this setting, given public coins, we obtain the following upper bound:

\begin{theorem}
    \label{thm:intro-public-upper}
    Let $1 \leq n_q \leq \lfloor \log_2(d) \rfloor$, and $\eps,\delta \in (0,1]$. In the distributed $(0,n_q,\mathsf{public},0)$-setting, the complexity of $(\eps,\delta)$-certification of $d$-dimensional states is at most $\bigo\Bigl(\frac{d^2}{2^{n_q} \eps^2}\Bigr)$. 
\end{theorem}

Note that when $2^{n_q} \geq \Omega(d)$, we achieve an $\bigo(d/\eps^2)$ upper bound, which is tight for state certification even in the centralized setting. Further, comparing \Cref{thm:intro-public-upper} to the $\bigo\Bigl(\frac{d^2}{2^{n_c/2} \eps^2}\Bigr)$ upper bound for state certification with public coins and limited \emph{classical} communication, we see that quantum communication provides a quadratic advantage over classical communication for this problem.

Moreover, we show that under a natural assumption, the above upper bound is tight. In particular, assume that the action of each distributed node is \emph{mixedness preserving}, i.e., when node $N_i$ receives the maximally mixed state $\rho = \frac{\mathbbm{1}_d}{d}$, its message to $N_c$ is the maximally mixed state of dimension $d_q \triangleq 2^{n_q}$, i.e., $\frac{\mathbbm{1}_{d_q}}{d_q}$. Then, we show the following lower bound:

\begin{theorem}
\label{thm:intro-public-lower}
     Let $d \geq 6, 1 \leq n_q \leq \lfloor \log_2(d) \rfloor$ and $0 < \eps < c$, where $c < 1$ is an absolute constant. In the distributed $(0,n_q,\mathsf{public},0)$-setting, assuming the action of each distributed node is mixedness preserving, the complexity of $\eps$-certification of $d$-dimensional states is at least $\Omega\Bigl(\frac{d^2}{2^{n_q} \eps^2}\Bigr)$.
\end{theorem}

Note that this mixedness-preserving assumption is indeed satisfied by our algorithm that achieves \Cref{thm:intro-public-upper}. Thus, \Cref{thm:intro-public-upper,thm:intro-public-lower} imply that, in the public-coin setting, $\Theta\left(\frac{d^2}{2^{n_q} \eps^2}\right)$ is the optimal complexity for state certification with mixedness-preserving strategies.

In the private-coin setting, we show a stronger lower bound, implying a strict separation between the private- and public-coin settings.

\begin{theorem}
\label{thm:intro-private-lower}
     Let $d \geq 6, 1 \leq n_q \leq \lfloor \log_2(d) \rfloor$ and $0 < \eps < c$, where $c < 1$ is an absolute constant. In the distributed $(0,n_q,\mathsf{private},0)$-setting, assuming the action of each distributed node is mixedness preserving, the complexity of $\eps$-certification of $d$-dimensional states is at least $\Omega\Bigl(\frac{d^3}{4^{n_q} \eps^2}\Bigr)$.
\end{theorem}

Given our lower bounds, one observes a strong dimension-dependent gap between the public- and private-coin settings (for constant $n_q$, this is of order $d$), highlighting that shared randomness is a powerful \emph{resource} in distributed learning and testing, as also observed in the settings of classical communication (see \Cref{tab:results}). In the resource-driven mindset of our framework, this is particularly interesting because it shows that the \emph{classical} resource of \emph{shared randomness} can be leveraged to reduce the use of another, often more expensive, resource: quantum communication. Notably, although the shared random seed is independent of the hypothesis state, it still leads to a substantial reduction in the quantum communication required for distributed certification. 

Finally, while our lower bounds hold only for mixedness preserving strategies, we conjecture that these lower bounds are tight, in both the public- and private-coin settings, for \textit{all} strategies. 

\begin{table}
    \centering
    \begin{tabular}{|c|c|c|c|}
    \hline
         & Distribution identity testing & \multicolumn{2}{|c|}{Quantum state certification} \\
    \hline
         & $2^{n_c} \leq d, n_q = 0, E = 0 $ & $2^{n_c} \leq d, n_q = 0, E = 0 $ & $n_c = 0, 2^{n_q} \leq d, E = 0 $ \\
    \hline
    Public-coin & $\Theta\left(\frac{d}{\sqrt{2^{n_c}} \eps^2}\right)$  \cite{acharya2020inferenceinformationconstraints, ACT20b} & $\Theta\left(\frac{d^2}{\sqrt{2^{n_c}}}\right)$ \cite{liu2024quantum} & $\Theta\left(\frac{d^2}{2^{n_q}\eps^2}\right)^\dag$ \Cref{thm:intro-public-upper,thm:intro-public-lower} \\
    \hline
    Private-coin & $\Theta\left(\frac{d^{3/2}}{2^{n_c} \eps^2}\right)$  \cite{acharya2020inferenceinformationconstraints, ACT20b} & $\Theta\left(\frac{d^3}{2^{n_c}}\right)$ \cite{liu2024quantum} &  $\Omega\left(\frac{d^3}{4^{n_q}\eps^2}\right) ^\dag$ \Cref{thm:intro-private-lower} \\
    \hline
    \end{tabular}
    \caption{Summary of our results and prior work in the distributed setting. $\dag$ indicates lower bounds holding only against mixedness-preserving strategies.}
    \label{tab:results}
\end{table}

\subsection{Technical overview}\label{ss:technical-overview}
Our main results are stated in the setting of limited quantum communication, and without any classical communication or shared entanglement. Thus, for ease of notation, we may drop the $(n_c,n_q,R,E)$-tuple notation when clear from context. In the sequel, unless specified otherwise, the public-coin setting refers to the $(n_c = 0, n_q, \mathsf{public}, E = 0)$-setting, and similarly the private-coin setting refers to $(n = 0, n_q, \mathsf{private}, E = 0)$, for some $1 \leq n_q \leq \lfloor \log_2(d)\rfloor$. More generally, instead of an $n_q$-qubit message, we assume that each distributed node sends a $d_q$-dimensional qudit, for some $2 \leq d_q \leq d$.

Now, in both the public- and private-coin settings, for $i \in [m]$, each node $N_i$ receives one copy of an unknown $d$-dimensional state $\rho$ and sends a $d_q$-dimensional qudit to the central node $N_c$. Thus, without loss of generality, we can model the action of $N_i$ as a quantum channel $\Phi_i : \mathbb{C}^{d \times d} \mapsto \mathbb{C}^{d_q \times d_q}$. In the private-coin setting, each node draws its channel independently of the others. However, in the public-coin setting, we imagine that some random string $\boldsymbol{r} \in \{0,1\}^*$ is known to all distributed nodes $N_i$ as well as the central node, and each channel $\Phi_i$ is parameterized by~$\boldsymbol{r}$.

Thus, we can model the setting as follows: the central node $N_c$ receives states $\Phi_1(\rho), \dots, \Phi_m(\rho)$, and then chooses an appropriate algorithm to perform a test on these post-processed states. Note that we place no memory restrictions on the central node, so we can imagine that this central algorithm performs fully coherent measurements on the post-processed states. This perspective on our framework will be helpful both in the design of the algorithm achieving our upper bound, as well as in the proofs of our lower bounds.

\subsubsection{Upper bounds}\label{sss:upper-bounds-overview}

Let us now present the main ideas behind our upper bound. First, note that by the data processing inequality, no matter the choice of channels $\{\Phi_i\}_{i \in [m]}$, for all $\rho \neq \sigma$, we will have
\begin{equation}
    \|\Phi_i(\rho) - \Phi_i(\sigma)\|_1 \leq \|\rho - \sigma\|_1.
\end{equation}
In other words, distinguishing $\rho$ from $\sigma$ only becomes harder in this setting. However, if we could show quantitative lower bounds on the distance between the states $\Phi_i(\rho)$ and $\Phi_i(\sigma)$ for some channels $\Phi_i$, i.e., bounds of the form
\begin{equation}
    \|\Phi_i(\rho) - \Phi_i(\sigma)\|_1 \geq C \cdot \|\rho - \sigma\|_1, 
\end{equation}
for an appropriate factor $C < 1$, then the central node could apply the tight state certification algorithm of \cite{buadescu2019quantum} to test with precision $C \cdot \eps$.  Thus, we wish to find channels $\{\Phi_i\}_i$ that maximize this distance preservation factor $C$. Instead of the trace distance, it will be easier to prove this for the Hilbert--Schmidt distance. 

Such distance-preserving dimensionality reduction maps have been studied previously in various settings. For the case of vectors in an $\ell_2$-space, this is the celebrated Johnson--Lindenstrauss lemma (see e.g., \cite[Section 5.3]{vershynin2018high}). Acharya, Canonne, and Tyagi \cite{ACT20b} also showed such a result for compressing probability distributions while preserving their $\ell_2$-distance, allowing them to achieve their tight public-coin upper bound for distribution identity testing with limited \emph{classical} communication. In the quantum setting, such a result is also known for \emph{pure} states \cite{sen2018quantum}, which was an important component of an algorithm for distributed inner product estimation with limited quantum communication \cite{arunachalam2025generalized}. Note that all of the bounds mentioned here are probabilistic, and are achieved by conceptually simple random operations. For instance, the pure state compression bound of \cite{sen2018quantum} is achieved by projecting states into random subspaces. Let us also note that for mixed states, Harrow, Montanaro, and Short \cite{harrow2015limitations} have already shown \emph{lower bounds} for distance-preservation with respect to the Hilbert--Schmidt and trace distances, as well as an upper bound for the trace distance. In particular, their Hilbert--Schmidt lower bound suggests that no algorithm can achieve a better factor than $C \leq \bigo\Bigl(\sqrt{\frac{d_q}{d}}\Bigr)$ with constant probability. 

We note that a generalization of the dimension-reduction map of \cite{arunachalam2025generalized} only achieves a distance-preservation factor of $C \geq \Omega(d_q/d)$ for mixed states with respect to the Hilbert--Schmidt distance. However, we will show that a remarkably simple channel achieves quadratically better distance-preservation factors, matching the lower bound of \cite{harrow2015limitations}. We will use the channel used by \cite{harrow2015limitations} for their trace-distance upper bound: simply apply a random unitary $\bfU$ to the state $\rho$, and trace out all but a $d_q$-dimensional subspace. In particular, letting $A,B$ be a fixed bipartition of $\mathbb{C}^{d \times d}$ with $A \cong \mathbb{C}^{d_q \times d_q}$, we define
\begin{equation}
    \Phi_{\bfU}(\rho) \triangleq \tr_{B}(\bfU \rho \bfU^\dag).
\end{equation}
Then, we show that the above channel satisfies
\begin{equation}
    \|\Phi_{\bfU}(\rho) - \Phi_{\bfU}(\sigma)\|_2 \geq \frac12 \sqrt{\frac{d_q}{d}} \cdot \|\rho - \sigma\|_2, 
    \label{eq:overview-1}
\end{equation}
with at least constant probability. We prove this inequality by using Weingarten calculus to appropriately bound the second and fourth moments of $\|\Phi_{\bfU}(\rho) - \Phi_{\bfU}(\sigma)\|_2$, and then applying the Paley-Zygmund inequality. We state this result in \Cref{lem:mixed-state-compression} and expect it to be of independent interest.

Given \Cref{eq:overview-1}, our algorithm is now straightforward: the distributed nodes use their shared randomness to sample a random unitary $\bfU$, and each node sends the state $\Phi_{\bfU}(\rho)$ to the central node $N_c$. $N_c$ then applies the optimal Hilbert--Schmidt state certification algorithm (see \Cref{lem:hscertify}) to these copies for an appropriate precision parameter $\eps^\prime$. This allows one to succeed at certification with constant probability, which can be boosted to an arbitrary $1-\delta$ probability by repeating this entire operation across $\bigo(\log(1/\delta))$ batches of distributed nodes. We refer to \Cref{s:public-coin-upper-bounds} for the complete proof of \Cref{thm:intro-public-upper} using this algorithm.

\subsubsection{Lower bounds}\label{sss:lower-bounds-overview}

We will now provide an overview of our lower bound techniques, and refer to \Cref{s:lower-bound-framework} for further details. Note that lower bounds for state certification are typically shown by proving the hardness of distinguishing between the maximally mixed state and a distribution over states that are $\eps$-far from the maximally mixed state. Clearly, any algorithm for state certification must be able to solve this point-versus-mixture distinguishing task, and so a lower bound for this latter task implies one for state certification. Similar ideas are also used to prove lower bounds for distribution identity testing, where one considers distinguishing between the uniform distribution and a mixture over distributions that are $\eps$-far from uniform. Before moving on to our techniques, let us describe prior lower bound techniques used by \cite{acharya2020inferenceinformationconstraints,liu2024quantum} to prove such lower bounds in the setting of classical communication.

\paragraph{Prior lower bounds for classical communication:} 
For both distribution testing and quantum state testing with limited classical communication, the central node $N_c$ receives a series of classical messages from the distributed nodes. To prove lower bounds for the distinguishing tasks above, it then suffices to show that the induced distributions over messages received by the central node $N_c$ in the two cases are statistically indistinguishable unless the number of nodes $m$ is sufficiently large. For any fixed input, as the distributed nodes cannot communicate with each other, the distribution over messages is easily shown to be a product distribution. Thus, to prove lower bounds for the point-versus-mixture task outlined above, one aims to show that a product distribution and a mixture of product distributions are statistically close. For such distributions, the Ingster--Suslina method (see e.g., \cite[Lemma 3.1]{CanonneTopicsDT2022}) allows one to easily upper bound the $\chi^2$-divergence between them. Indeed, \cite{acharya2020inferenceinformationconstraints,liu2024quantum} both use this method to upper bound the $\chi^2$-divergence between the induced distributions as a function of the number of nodes $m$ and the communication channels implemented by each node. Standard arguments imply that this divergence must be large for these distributions to be statistically distinguishable. Thus, these upper and lower bounds on the $\chi^2$-divergence together yield a lower bound on $m$.

While the above arguments allow \cite{acharya2020inferenceinformationconstraints,liu2024quantum} to obtain tight lower bounds in the public-coin setting, they do not immediately allow one to prove stronger lower bounds in the private-coin setting. To extend to the private coin setting, Acharya, Canonne, and Tyagi \cite{acharya2020inferenceinformationconstraints} essentially showed that, here, one can imagine that the mixture of alternatives was chosen adversarially to minimize the $\chi^2$-divergence \emph{conditioned} on a fixed choice of channels implemented by the distributed nodes. This allows them to prove tighter upper bounds on the $\chi^2$-divergence, in turn leading to stronger lower bounds in the private-coin setting.

\paragraph{Lifting to quantum communication:}

Recall that in our setting of quantum communication, on input state $\rho$, the central node $N_c$
receives quantum messages $\Phi_1(\rho), \dots, \Phi_m(\rho)$. Thus, to solve the distinguishing task, the central node must be able to distinguish between the two states $\bigotimes_{i = 1}^m \Phi_i\Bigl(\mmstate\Bigr)$ and $\mathbb{E}_{\rho \sim D}  [ \bigotimes_{i = 1}^m \Phi_i(\rho)]$. Thus, in analogy with the setting of classical communication, we aim to show that the former product state and the latter mixture over product distributions are statistically indistinguishable unless $m$ is large. To prove this, we will make use of the recently developed quantum analogue of the Ingster--Suslina method \cite{odonnell2025instanceoptimalquantumstatecertification}, allowing us to upper bound the \emph{quantum} $\chi^2$-divergence between such states; see \Cref{lem:quantum-ingster-suslina} for a precise statement of this method. It is also easy to show that for the two states to be statistically indistinguishable, this quantum $\chi^2$ divergence must be larger than a constant. Again, these upper and lower bounds together imply a lower bound on the number of distributed nodes.

Further, in the private-coin setting, we also lift the arguments of \cite{acharya2020inferenceinformationconstraints}, showing that any successful private-coin protocol must satisfy
\begin{equation}
    \max_{\Phi_1, \dots, \Phi_m} \min_{D} \qdchi\left(\mathbb{E}_{\rho \sim D}  \left[ \bigotimes_{i = 1}^m \Phi_i(\rho)\right] \Bigg\| \bigotimes_{i = 1}^m \Phi_i\Bigl(\mmstate\Bigr)\right) \geq \frac{1}{16},
\end{equation}
where we minimize over all distributions $D$ over states that are $\eps$-far from the maximally mixed state. In other words, in correspondence with the classical setting, we can imagine the mixture of alternatives is picked adversarially to minimize the quantum $\chi^2$-divergence \emph{after} the channels implemented by the nodes have been chosen. This will later allow us to prove our improved lower bound in the private-coin setting.

\paragraph{Hard mixture of alternatives:} We will use the mixture of alternatives previously introduced by Liu and Acharya \cite{liu2024role} and later used by them in the setting of limited classical communication~\cite{liu2024quantum}. In particular, they pick $\ell$ orthonormal traceless matrices $V_1, \dots, V_\ell$ in $\mathbb{C}^{d \times d}$, and randomly perturb the maximally mixed state along these directions. In particular, for $\bfz \sim \{-1,+1\}^\ell$, a vector consisting of $\ell$ uniformly random Rademacher variables, they define
\begin{equation}
    \rho_{\bfz} \approx \mmstate + \frac{c\eps}{\sqrt{d\ell}} \sum_{i = 1}^\ell V_i z_i.
\end{equation}
Applying the quantum Ingster--Suslina method (\Cref{lem:quantum-ingster-suslina}) to this hard instance, we wish to bound the moment generating function of the quantity $Z_{\Phi}(\bfz,\bfz^\prime)$ with respect to $\bfz,\bfz^\prime \sim \{-1,+1\}^\ell$, where we have
\begin{equation}
    Z_{\Phi}(\bfz,\bfz^\prime) = \tr\left(\Phi\Bigl(\mmstate\Bigr)^{-1} \Phi(\Delta_{\bfz}) \Phi(\Delta_{\bfz^\prime})\right),
\end{equation}
and $\Delta_z = \rho_z - \mmstate$.
To handle the $\Phi(\mathbbm{1}/d)^{-1}$ term appearing above, we invoke our assumption that the channels used by the distributed nodes are \emph{mixedness-preserving}, i.e., that $\Phi(\mathbbm{1}/d) = \mathbbm{1}_{d_q}/d_q$. Under this assumption, $Z_{\Phi}(\bfz,\bfz^\prime)$ can be rewritten as a quadratic polynomial in the entries of the random Rademacher vectors $\bfz,\bfz^\prime$. We then use a standard upper bound on the moment generating functions of such quadratic forms (see \Cref{lem:mgf-quadratic-form}), allowing us to obtain a lower bound on $m$ that depends on the basis $V_1, \dots, V_\ell$ and the channels $\Phi_1, \dots, \Phi_m$.

In particular, let $\mc{V} = [\mathrm{vec}(V_1), \dots, \mathrm{vec}(V_\ell)]$, and $T(\Phi_1, \dots, \Phi_m) = \frac{1}{m} \sum_{i = 1}^m M_{\Phi_i}^\dag M_{\Phi_i}$, where $M_\Phi$ is the Liouville matrix representation of a channel $\Phi$ (see \Cref{s:preliminaries}). Then, the above arguments allow us to show a lower bound scaling inversely with the norm $\|\mc{V}^\dag T(\Phi_1, \dots, \Phi_m) \mc{V}\|_2$. In the public-coin setting, we can only prove a worst-case upper bound on this quantity. However, in the private-coin setting, we can adversarially choose the perturbations $V_1, \dots, V_\ell$ to minimize the norm above. In either case, we are led to upper bounding Schatten norms of $\{M_{\Phi_i}\}_i$. We provide generic upper bounds on these quantities in \Cref{ss:norm-bounds}, finally yielding our lower bounds in \Cref{thm:intro-public-lower,thm:intro-private-lower}. We refer to \Cref{s:lower-bounds} for complete proofs of our lower bounds.

We note that for $d_q = d$, our lower bounds recover the centralized $\Omega(d/\eps^2)$ lower bound for state certification~\cite{o2015quantum}. Moreover, in \Cref{s:centralized-mixedness-testing-bound}, we also apply the quantum Ingster--Suslina method of \cite{odonnell2025instanceoptimalquantumstatecertification} to the hard mixture of \cite{acharya2020inferenceinformationconstraints}, obtaining an alternate self-contained proof of this $\Omega(d/\eps^2)$ lower bound. 

\subsection{Prior work}\label{ss:prior-work}

There is a wide array of prior work, both classical and quantum, which is relevant to this work. 

\textbf{Classical distributed statistical inference with communication constraints:} There is a rich history of prior work on classical distributed statistical inference under communication constraints, in a wide variety of communication models, and with differing objectives. Following from initial work in this direction~\cite{AhlswedeCsiszar1986, Han1987, HanAmari1998}, one well studied model is the one in which all samples are available to one party, which then communicates via some constrained communication model~\cite{, XiangKim2013, WiggerTimo2016, SahasranandTyagi2018, AndoniMalkinNosatzki2019}. However, the classical works most relevant to ours, are those in which each sample is initially held by a separate node in a constrained communication network.

In such a model, the first problem to be considered was distributed mean estimation~\cite{ZhangDuchiJordanWainwright2013, GargMaNguyen2014, Shamir2014, BravermanGargMaNguyenWoodruff2016, XuRaginsky2017}. This was followed by a growing body of work on distributed learning and testing of distributions under communication constraints. Some early works on this topic, on both distribution learning~\cite{DiakonikolasGrigorescuLiNatarajanOnakSchmidt2017} and distribution testing~\cite{DiakonikolasGouleakisKaneRao2019}, assumed a blackboard model of communication -- in which all messages are visible to all parties -- and aimed to characterize the \textit{total} amount of communication required. The classical version of the model we consider, however, was first studied in Refs.~\cite{HanMukherjeeOzgurWeissman2018,acharya2020inferenceinformationconstraints}. The most relevant of these to our work is Ref.~\cite{acharya2020inferenceinformationconstraints}, which as previously mentioned, introduced both an abstract SMP model for distributed inference under communication constraints, with each sample held by a different node, and provided a methodology for proving lower bounds in such a model. Indeed, as mentioned in Section~\ref{ss:framework}, our model is a natural quantum generalization of the model introduced in~\cite{acharya2020inferenceinformationconstraints}, and as discussed in Section~\ref{ss:technical-overview}, our lower bounds are obtained by techniques directly inspired by the ones introduced in this work.

Following the introduction of the multi-node constrained SMP framework in Ref.~\cite{acharya2020inferenceinformationconstraints}, multiple works then provided specific techniques and algorithms for learning and testing distributions in this model~\cite{ACT19,ACT20b,ACHST20, ACT20c}, in the process characterizing the interplay between public/private randomness and information constraints. Simultaneously, Ref.~\cite{acharya2021inferenceinformationconstraintsiii} considered the setting of \textit{privacy-preserving} communication channels, and gave algorithms for learning and testing in this setting. Interactivity was then introduced into the model and studied for learning and testing discrete distributions~\cite{ACLST22} and high-dimensional continuous parameteric distributions~\cite{ACST23}, and for non-parameteric density estimation~\cite{ACST24}. Additionally, more recent work has again considered the setting of multiple samples per node~\cite{ACLST21,Vuursteen2024}, which is a natural abstract model for federated learning. Finally, we note that recent work has also studied \textit{hypothesis testing} in the distributed model with communication constraints~\cite{PJL24,PAJL23,PAJL25,KPJ25}.

\textbf{Quantum learning and testing:} Recent years have witnessed a large amount of work on the problem of learning and testing quantum states, processes and systems, in a wide variety of learning models, under the broad umbrella of \textit{quantum learning theory}. Providing a complete survey of such work is not possible here, but we refer to reviews on the complexity of learning quantum states~\cite{anshu2023surveycomplexitylearningquantum} and quantum property testing~\cite{gs007}, as well as the recently established \textit{quantum learning theory zoo}~\cite{QuantumLearningTheoryZoo}, which provides an up-to date repository of work on quantum learning theory. We note that the sub-field of quantum learning theory most relevant to our work is that of learning and testing \textit{quantum states} from multiple copies. Indeed, while we focus primarily on quantum state certification, any problem of this type could be immediately studied in the distributed model we introduce here. 

\textbf{Centralized quantum state certification:} In the centralized setting, initial work on quantum state certification focused on characterizing the worst-case complexity of the task both with~\cite{o2015quantum,buadescu2019quantum} and without~\cite{bubeck2020entanglement} the ability to perform coherent measurements on multiple copies of the unknown quantum state. More recent work has provided \textit{instance optimal} bounds in both settings~\cite{chen2022toward,chen2022tightStateCertification, odonnell2025instanceoptimalquantumstatecertification} and shown that without coherent measurements, adaptivity provides no advantage for this problem. Additionally, the problem of state certification with incoherent measurements restricted to a certain number of outcomes was recently studied in Ref.~\cite{liu2024quantum}. As mentioned earlier, this setting is equivalent to the $(n_c\leq n,n_q=0,R,E=0)$ version of the model we study here, and recovers the worst-case results of~\cite{bubeck2020entanglement,chen2022toward} when $n_c=n$.  Finally, recent work has studied the role of \textit{shared randomness} in quantum state certification with incoherent measurements~\cite{liu2024role}, contrasting the complexity of this task in the public and private coin settings, as well as the case of \textit{non-iid} samples~\cite{depalma2025noniidhypothesistestingclassical}.

\textbf{Distributed inner product estimation (DIPE):} Given copies of two quantum states $\rho$ and $\sigma$, the problem of estimating $\tr(\rho\sigma)$ is known as inner product estimation, and provides a fundamental quantum algorithmic primitive and one quantum generalization of distribution closeness testing. In the centralized setting (with multi-copy coherent measurements), the SWAP test provides an algorithm with only constant sample complexity. Motivated by similar reasoning as us, recent work has characterized the sample complexity of  \emph{distributed} inner product estimation -- in this setting, one party holds copies of $\rho$ and the other party holds copies of $\sigma$ -- with both only classical communication~\cite{anshu2022distributed} and limited quantum communication~\cite{gong2024samplecomplexitypurityinner,arunachalam2025generalized} allowed between parties. In the distributed setting, this problem has also been called \textit{cross device verification} for its applications in the verification of quantum devices~\cite{EfficientDIPE2}. More recent works have also focused on characterizing the class of quantum states for which \textit{computationally} efficient DIPE is possible~\cite{EfficientDIPE1,EfficientDIPE2}. This line of work is very similar in spirit to the work on distributed state certification we initiate here, with the primary differences being (a) the choice of testing problem, (b) the fact that the input to DIPE is multiple copies of \textit{two} unknown states, and (c) the fact that the unknown states are distributed among only two nodes as opposed to $m$ distributed nodes. Additionally,~\cite{arunachalam2025generalized} consider only pure states, and while~\cite{gong2024samplecomplexitypurityinner} consider mixed states, they also allow one-way \textit{classical} communication in addition to limited quantum communication. We discuss how to study such two-state inference problems in our framework in Section~\ref{ss:discussion},  and note that the exact complexity of such problems in our model is not immediately implied by any of these prior results.

\subsection{Discussion}\label{ss:discussion}

In this work, we have introduced a new model for distributed quantum inference by generalizing the classical setup of \cite{acharya2020inferenceinformationconstraints} and have studied the problem of state certification in the public- and private-coin settings. Our lower bounds only hold under the assumption that the strategies of the distributed nodes are mixedness-preserving. While we conjecture that our bounds are tight even without this assumption, proving this remains an important open question. As mentioned in \Cref{ss:technical-overview}, the key difficulty in getting around this assumption is the inverse term appearing in the quantum $\chi^2$-divergence. It is unclear whether one can prove such lower bounds against all strategies while still bounding the quantum $\chi^2$-divergence, or if entirely novel techniques would be necessary for this task. On the upper bound side, we have only shown an upper bound in the public-coin setting, and do not obtain any upper bounds in the private-coin setting. The private-coin distribution testing algorithm of \cite{ACT20b} makes use of a distributed simulation primitive, which allows the central node to simulate a single sample from the original distribution given messages from multiple distributed nodes, allowing algorithms for both testing and learning distributions. However, this distributed simulation algorithm does not seem to generalize to the quantum setting, and it is unclear how to otherwise obtain a private-coin upper bound for certification in our framework.

Apart from these immediate open questions, there is a vast array of future directions one can consider within our distributed inference framework. Our results and those of prior work only hold for the case of either limited quantum communication or limited classical communication; could we allow \emph{both} classical and quantum communication, and show smooth bounds as functions of the parameters $n_c,n_q$? In such a setting, it would even be interesting to understand the complexity when $n_c$ is unbounded while still keeping $n_q$ limited. Next, note that we have only considered the problem of state certification in this work. However, in the analogous classical setting, the complexity of distribution learning has also been characterized. Similarly, in our distributed setting, could we fully characterize the complexity of state tomography? Lastly, our framework also allows the distributed nodes to share entanglement. In \Cref{s:entanglement}, we demonstrate a class of problems, including purity testing~\cite{chen2022exponential} and unsigned Pauli shadow tomography~\cite{chen2024optimalA}, that require exponentially many copies in our setting even with unbounded classical communication, but which can be solved with a constant number of copies if shared entanglement is permitted. This positions entanglement as a powerful resource in our distributed framework, and it would be interesting to see how the presence of this resource affects the complexity of other inference tasks.

Moreover, motivated by classical follow-ups to the work of \cite{acharya2020inferenceinformationconstraints}, one can consider many interesting extensions to our framework. We have only considered the simultaneous message-passing model, where distributed nodes are only allowed one-way communication to the central node. It would be interesting to model more general settings by understanding the effects of two-way communication and inter-node communication on the complexity of state certification. Already, we note that allowing $n$-qubits of quantum communication among  $m'$ distributed nodes is essentially equivalent to grouping those nodes into one and allowing an $m'$-copy measurement. This allows one to recover the setting of $m'$-copy memory in the distributed nodes, for any $m'$. Similarly, instead of allowing only one copy per distributed node, we could imagine that each node obtains multiple copies of the unknown state. For our public-coin results, we have assumed that the nodes share an unbounded amount of randomness. From a practical perspective, an important direction of study would thus be to quantify the effects of limited shared randomness on the complexity of state certification. As we move towards future deployments of such distributed quantum settings, it would also be important to understand the tradeoffs imposed by requiring the use of privacy-preserving algorithms, and the requirements necessary to learn and test in the presence of \textit{untrusted} parties in the network.

Lastly, we note that we have only discussed the study of \emph{one-state} inference problems. More generally, one can imagine that the nodes receive copies of distinct states. For instance, if $m/2$ distributed nodes receive copies of a state $\rho$ and the other $m/2$ nodes receive copies of another state $\sigma$, one can now consider testing for properties of both states. This setting would allow modeling the problem of DIPE discussed above, as well as the generalization of state certification to closeness testing, i.e., the problem of testing whether the two \emph{unknown} states $\rho,\sigma$ are identical or $\eps$-far in trace distance. We remark that our algorithm for state certification in \Cref{thm:intro-public-upper} also already works for this problem of closeness testing with just some minor adjustments. However, for (mixed state) DIPE in our setting, we note that the channel we used for our state certification upper bound does not seem to obtain desirable rates, and neither does the channel used by~\cite{arunachalam2025generalized} for DIPE of pure states with limited quantum communication. Additionally, the algorithm of \cite{gong2024samplecomplexitypurityinner} for mixed-state DIPE also uses classical communication, and this may be necessary in our setting as well. This points to two natural open directions: determining the optimal compression channel for inner product estimation without classical communication, and developing a distributed model with controlled inter-node communication that bridges our one-copy-per-node setting and the existing two-party frameworks. Now, generalizing this two-state setting further, one can imagine that \emph{each} node receives a distinct state, i.e., we have states $\rho_1, \dots, \rho_m$. Here, we may consider testing properties of the average of these $m$ states, in the spirit of the non-iid setting of ~\cite{garg2023testing,depalma2025noniidhypothesistestingclassical}.

\subsection*{Organization}
We start by presenting preliminary background and necessary notation in \Cref{s:preliminaries}. We prove our upper bound, \Cref{thm:intro-public-upper}, in \Cref{s:public-coin-upper-bounds}. We provide further technical background for our lower bounds in \Cref{s:lower-bound-framework}, and prove our lower bound \Cref{thm:intro-private-lower,thm:intro-public-lower} in \Cref{s:lower-bounds}. \Cref{s:centralized-mixedness-testing-bound} contains a new self-contained proof of the $\Omega(d/\eps^2)$ lower bound for mixedness testing in the centralized setting using the hard instance considered in this work. Lastly, in \Cref{s:entanglement}, we demonstrate an exponential separation between algorithms possessing shared entanglement in the distributed setting and those without. 

\subsection*{Acknowledgments}
The authors thank Matthias Caro for helpful discussions on quadratic forms of Rademacher random variables. The authors acknowledge the use of Claude Opus 4.6 to aid with the proofs of \Cref{lem:unital-channel-operator-norm-upper-bound,lem:unital-channel-frobenius-norm-upper-bound}. MD and CW acknowledge the support of the Quantum Advantage Pathfinder (QAP), with grant reference EP/X026167/1, and the UK Engineering and Physical Sciences Research Council. RS is grateful to the Alexander von Humboldt Foundation for support under the German Research Chair program at the African Institutes for Mathematical Sciences. Part of this work was carried out while MD and CW visited the African Institute for Mathematical Sciences, Cape Town for the 1st AIMS Workshop on the Theory of Quantum Learning Algorithms (2025).

\section{Preliminaries and notation}\label{s:preliminaries}

We start by providing background on some basic notions of quantum information, and refer to \cite{watrous2018theory} for further details. A $d$-dimensional quantum state is represented by a density matrix, i.e., a complex $d \times d$ positive semi-definite operator with trace $1$. Let $\mathbbm{1}_d \in \mathbb{C}^{d \times d}$ be the identity operator. We may omit the subscript when the dimension is clear from context. A special state of interest is the maximally mixed state, i.e., $\mmstate$.

A quantum channel is a completely positive and trace-preserving linear map. In general, the dimensions of the input and output spaces of a channel need not be identical. Let us now define two classes of quantum channels that will be of particular importance.
\begin{definition}
    \label{def:unital-and-mixedness-preserving-maps}
    A map $\Phi : \mathbb{C}^{d \times d} \mapsto \mathbb{C}^{d^\prime \times d^\prime}$ is said to be \emph{unital} if $\Phi(\mathbbm{1}_d) = \mathbbm{1}_{d^\prime}$. Further, we say a quantum channel $\Phi : \mathbb{C}^{d \times d} \mapsto \mathbb{C}^{d^\prime \times d^\prime}$ is mixedness preserving if $\Phi(\mathbbm{1}_d/d) = \mathbbm{1}_{d^\prime}/d^\prime$.
\end{definition}
 Note that the sets of unital maps and mixedness-preserving quantum channels are distinct whenever $d \neq d^\prime$. Indeed, in such cases, a unital map cannot even be a quantum channel as it is clearly not trace-preserving.

Let us now discuss some ways of representing quantum states and channels. First, we define the \emph{vectorization map} $\mathrm{vec} : \mathbb{C}^{d \times d} \mapsto \mathbb{C}^{d^2}$ which rearranges the entries of a matrix into a vector. Note that the vectorization operation is inner-product preserving, i.e., $\langle X, Y \rangle = \langle \mathrm{vec}(X), \mathrm{vec}(Y) \rangle$, where the operator inner product is the Hilbert--Schmidt inner product, $\langle X, Y\rangle \triangleq \tr(X^\dag Y)$. Now, to represent the action of a channel in the vectorized picture, we make use of its \emph{Liouville matrix} representation. Let $\Phi: \mathbb{C}^{d \times d} \mapsto \mathbb{C}^{d^\prime \times d^\prime}$. Then, the associated Liouville matrix $M_\Phi \in \mathbb{C}^{d^{\prime 2} \times d^2}$ is defined as the unique matrix satisfying $\mathrm{vec}(\Phi(X)) = M_\Phi \cdot \mathrm{vec}(X)$, for all $X \in \mathbb{C}^{d \times d}$.

Quantum channels can also be represented by their Kraus operator decomposition. In particular, any completely positive superoperator $\Phi: \mathbb{C}^{d \times d} \mapsto \mathbb{C}^{d^\prime \times d^\prime}$ can be described by a set of Kraus operators $\{A_k\}_k \subseteq \mathbb{C}^{d^\prime \times d}$, such that the action of the channel is given by $\Phi(X) = \sum_k A_k X A^\dag_k$. For quantum channels, the trace-preserving condition implies $\sum_k A_k^\dag A_k = \mathbbm{1}_d$. 

We will also make use of the \emph{Choi representation} of channels. Let $A,B$ denote two $d$-dimensional registers. Then, for a channel $\Phi: \mathbb{C}^{d \times d} \mapsto \mathbb{C}^{d^\prime \times d^\prime}$, its Choi representation $J(\Phi) \in \mathbb{C}^{d \times d} \otimes \mathbb{C}^{d^\prime \times d^\prime}$ is given by 
\begin{equation}
    J(\Phi) = \mc{I}_A \otimes \Phi_B \left(\sum_{i,j = 1}^d \ket{i,i}_{A,B} \bra{j,j}_{A,B}\right),
\end{equation}
where $\mc{I}$ is the identity channel. It is a standard fact that $\tr_B(J(\Phi)) = \mathbbm{1}_d$ and $\tr(J(\Phi)) = d$. Note that for all completely positive maps $\Phi$, the Choi representation $J(\Phi)$ is positive semidefinite. Thus, its trace-normalized version $v(\Phi) \triangleq \frac{J(\Phi)}{d}$ is a quantum state, known as the \emph{Choi state}. It is also a standard fact that the Choi representation and the Liouville matrix of a channel are related to each other by a suitable rearrangement of entries.

Let us now state some useful distances between quantum states. For any $p \geq 1$, we will let $\|\cdot\|_p$ denote the Schatten $p$-norm. Then, the Hilbert--Schmidt distance between two states $\rho,\sigma$ is given by $\frac12\|\rho - \sigma\|_2$, and their trace distance is given by $\frac12 \|\rho - \sigma\|_1$. While we state our results with respect to trace-distance testing, the Hilbert--Schmidt distance will be easier to analyze for our upper bounds. For this purpose, we will also use the following centralized Hilbert--Schmidt testing upper bound from \cite{buadescu2019quantum}.

\begin{lemma}[Hilbert--Schmidt State Certification; {\cite[Theorem 1.4]{buadescu2019quantum}}]
\label{lem:hscertify}
    There is an algorithm $\mathsf{HSCertify}(\sigma,\epsilon,\delta)$ that takes $\bigo(\log(1/\delta)/\epsilon^2)$ copies of an unknown state $\rho \in \mathbb{C}^{d \times d}$, and can distinguish between the cases $\rho = \sigma$ and $\|\rho - \sigma\|_2 \geq \epsilon$ with probability at least $1-\delta$.
\end{lemma}

To prove our testing lower bounds, we will also make use of the quantum $\chi^2$-divergence between two states. Note that there are multiple definitions of such quantum divergences, but we only state the one that will be relevant to our analysis.

\begin{definition}[Quantum $\chi^2$-divergence]
\label{def:quantum-chi2}
    Let $\rho,\sigma \in \mathbb{C}^{d \times d}$. When $\mathsf{supp}(\rho) \subseteq \mathsf{supp}(\sigma)$, we define
    \begin{equation}
        \qdchi(\rho \| \sigma) = \tr(\sigma^{-1}(\rho-\sigma)^2) = \tr(\sigma^{-1}\rho^2) - 1;
    \end{equation}
    otherwise, we define it to be $\infty$.
\end{definition}

The trace distance and the quantum $\chi^2$-divergence are related by the following inequality (see e.g., \cite[Lemma 5]{temme2010chi}):

\begin{equation}
\label{eq:dtr-qdchi}
    \|\rho - \sigma\|_1 \leq \qdchi(\rho \| \sigma).
\end{equation}

\subsection{Moment and tail bounds}

We will now state some results allowing us to integrate over the unitary Haar measure $U(d)$, which will be important in the analysis of our upper bounds. In general, $k$-order moments of random unitaries can be related to the elements of the symmetric group $\mc{S}_k$, as formalized in the following standard lemma (see e.g., \cite{mele2024introduction}).

\begin{lemma}
\label{lem:weingarten-calc}
Given a permutation $\pi \in \mc{S}_k$, let $P_\pi \in \mathbb{C}^{d^k \times d^k}$ be the associated permutation operator and $\wg(\pi,d)$ be its Weingarten coefficient. Let $M \in \mathbb{C}^{d^k \times d^k}$. Then,
    \begin{equation}
        \mathbb{E}_{\bfU \sim U(d)}[\bfU^{\otimes k} M \bfU^{\dag \otimes k}] = \sum_{\pi,\tau \in \mc{S}_k} \wg(\pi^{-1}\tau,d)\tr(P_\tau^\dagger M) P_\pi.   
    \end{equation}
\end{lemma}

Note that we haven't defined the Weingarten coefficients mentioned above. These are real coefficients associated with each permutation and dependent on the dimension $d$. In our work, we will only be interested in computing second- and fourth-order moments, and state bounds on these Weingarten coefficients specifically.

\begin{lemma}[{\cite[Tables II \& III]{brouwer1996diagrammatic}}]
\label{lem:weingarten-coeffs}
    Let $\mathrm{id}_2, \mathrm{swap} \in \mc{S}_2$ be the identity and the swap permutations respectively. Then, 
    \begin{equation}
        \wg(\mathrm{id}_2,d) = \frac{1}{d^2-1}, \quad \wg(\mathrm{swap},d) = \frac{-1}{d(d^2-1)}.
    \end{equation}
    Further, let $\mathrm{id}_4$ be the identity permutation in $S_4$. Then, 
    \begin{equation}
        |\wg(\mathrm{id}_4,d)| \leq d^{-4}, \quad \text{and } \forall \text{ }\mathrm{id}_4 \neq \pi \in \mc{S}_4, \quad \wg(\pi,d) \leq \bigo(d^{-6}).
    \end{equation}
\end{lemma}

As a corollary, let us explicitly state a second-order integral formula.

\begin{corollary}[Second-order integral]
\label{cor:second-order}
For any $A_1,A_2,B_1,B_2 \in \mathbb{C}^{d \times d}$, we have
    \begin{align}
        &\mathbb{E}_{\bfU \sim U(d)}[\tr(\bfU^{\otimes 2} (A_1 \otimes A_2)\bfU^{\dag \otimes 2} (B_1 \otimes B_2))] \nonumber\\&= \frac{\tr(A_1A_2)\tr(B_1B_2) + \tr(A_1 \otimes A_2) \tr(B_1 \otimes B_2)}{(d^2-1)} - \frac{\tr(A_1A_2)\tr(B_1 \otimes B_2) + \tr(A_1 \otimes A_2)\tr(B_1B_2)}{d(d^2-1)}.
    \end{align}
\end{corollary}

Finally, in our lower bounds proofs, we will need to upper-bound MGFs of quadratic forms of vectors of random Rademacher variables. We state the following standard upper bound on such MGFs below: 
\begin{lemma}
\label{lem:mgf-quadratic-form}
    Let $\bfz,\bfz^\prime \sim \{-1,+1\}^{\ell}, \lambda \in \mathbb{R}, A \in \mathbb{R}^{\ell \times \ell}$. Then,
    \begin{equation}
        \mathbb{E}_{\bfz,\bfz^\prime}[\exp(\lambda \bfz^\top A \bfz^\prime)] \leq \exp\left(C\lambda^2\|A\|_F^2\right)
    \end{equation}
    whenever $\lambda \leq \frac{1}{2\|A\|_{\infty}}$, and where $C > 0$ is an absolute constant.
\end{lemma}

\begin{proof}
    By Lemmas 6.2.3 and 6.2.4 of \cite{vershynin2018high}, it suffices to show that $z,z^\prime$ have constant subgaussian norms (see \cite[Definition 2.6.4]{vershynin2018high}). By \cite[Exercise 2.24(c)]{vershynin2018high}, a Rademacher random variable has constant subgaussian norm, and by \cite[Lemma 3.4.2]{vershynin2018high}, so do vectors of independent Rademacher variables, as desired. 
\end{proof}

\section{Public-coin upper bound}\label{s:public-coin-upper-bounds}

In this section, we prove our upper bound for distributed state certification using only quantum communication and public coins, i.e., the $(n_c = 0, n_q, \mathsf{public}, E = 0)$ setting of \Cref{ss:framework}. More generally, instead of only allowing $n_q$-qubit messages, we will allow the algorithm to send a single qudit of dimension $d_q$, for any $2 \leq d_q \leq d$. We restate this upper bound in the following theorem: 
\begin{theorem}
\label{thm:public-upper-bound}
    Let $d \geq d_q \geq 2$, $\eps > 0.$ Consider the public-coin distributed setting with $m$ distributed nodes, each capable of sending a qudit of dimension $d_q$. Then, to $\eps$-certify a $d$-dimensional state $\sigma$ with probability at least $1-\delta$, it suffices to take $m \leq \bigo\left(\frac{d^2}{d_q \eps^2} \cdot \log(1/\delta)\right)$.
\end{theorem}

Our arguments to prove the above theorem assume that $d_q$ divides $d$, but these can easily be adjusted to hold for general $d_q$\footnote{When $d_q$ does not divide $d$, one can embed the unknown state and the hypothesis state into a space of dimension $d^\prime = d_q \times \lceil d/d_q \rceil$ while leaving their distance unchanged. This increases the dimension by at most a constant factor, as $d^\prime < d +d_q \leq 2d$, allowing our subsequent arguments to follow with only minor adjustments.}. Now, let $A,B$ be a bipartition of $\mathbb{C}^{d \times d}$ where the dimension of $A$ is $d_A = d_q$. For any unitary $U$, let $\Phi_U(\rho) \triangleq \tr_B(U\rho U^\dag)$. Before stating our algorithm formally, let us provide a brief overview of its steps. The central node $N_c$ and a batch of distributed nodes will sample a random unitary $\bfU$, and the distributed nodes will each send the $d_q$-dimensional state $\Phi_{\bfU}(\rho)$ to $N_c$. Then, $N_c$ will use \Cref{lem:hscertify} to test whether $\Phi_{\bfU}(\rho)$ and $\Phi_{\bfU}(\sigma)$ are identical or appropriately far in $\|\cdot\|_2$-distance with constant success probability. Repeating this across many batches allows us to boost the overall probability to $1-\delta$. The complete description of the algorithm is stated in \Cref{alg:public-coin-improved}.

\begin{algorithm}[H]
    \begin{algorithmic}[1]
        \caption{Distributed state certification with $d_q$-dimensional quantum communication and public coins}
        \label{alg:public-coin-improved}
        \State Set $R = \bigo(\log(1/\delta)), m^\prime = \bigo(d^2/d_q\epsilon^2)$, such that $m = m^\prime \cdot R$.
        \State Set $\epsilon^\prime, \delta^\prime,\tau$ according to \Cref{eq:parameters-2}.
        \State All distributed nodes and the central node use their shared randomness to sample $R$ random unitaries $\bfU_1,\dots,\bfU_R \sim U(d)$.
        \For{$r = 1$ to $R$}
            \State Nodes $N_{(r-1)\cdot m^\prime + 1},\dots,N_{r\cdot m^\prime}$ apply $\Phi_{\bfU_r}$ to their individual copies of $\rho$ and send the output state to $N_c$.
            \State The central node $N_c$ applies $\mathsf{HSCertify}(\Phi_{\bfU_r}(\sigma),\epsilon^\prime,\delta^\prime)$ to these states and records the outcome ``Close'' or ``Far''. 
        \EndFor
        \If{``Far'' occurs more than $\tau R$ times} reject.
        \Else{ } accept.
        \EndIf
    \end{algorithmic}
\end{algorithm}

\begin{remark}
    Note that the above algorithm involves unitaries drawn from the Haar measure over $U(d)$. As we will see below, our analysis only involves up to fourth-order moments of these random unitaries. Consequently, one can instead draw the unitaries from any unitary $4$-design.
\end{remark}

\subsection{Upper bound proof}

Our upper bound will leverage the fact that the channel $\Phi_{\bfU}$ defined above only compresses the distance between two states by a factor proportional to $\sqrt{\frac{d_A}{d}}$:

\begin{lemma}
\label{lem:mixed-state-compression}
    Let $\rho,\sigma \in \mathbb{C}^{d \times d}$ be two quantum states. Then, there exists an absolute constant $C_2 > 1$ such that
    \begin{equation}
        \Pr_{\bfU \sim U(d)}\left[\|\Phi_{\bfU}(\rho) - \Phi_{\bfU}(\sigma)\|_2 \geq \frac12\sqrt{\frac{d_A}{d}} \cdot \|\rho-\sigma\|_2\right] \geq \frac{1}{C_2}.
    \end{equation}
\end{lemma}

We will prove the above lemma in the subsequent section. Now, with \Cref{lem:mixed-state-compression}, we can prove \Cref{thm:public-upper-bound}:

\begin{proof}[Proof of \Cref{thm:public-upper-bound}]
    Consider a fixed iteration $r \in [R]$ of \Cref{alg:public-coin-improved}. In the case of $\rho \neq \sigma$, \Cref{lem:mixed-state-compression} says 
    \begin{align}
     \Pr\left[\|\Phi_{\bfU_r}(\rho) - \Phi_{\bfU_r}(\sigma)\|_2 \geq \frac12\sqrt{\frac{d_A}{d}} \cdot \|\rho-\sigma\|_2\right] \geq \frac{1}{C_2}.
    \end{align}
    When $\|\rho-\sigma\|_1 \geq \epsilon$, a standard norm inequality implies $\|\rho - \sigma\|_2 \geq \epsilon/\sqrt{d}$, and therefore with probability at least $1/C_2$, we have 
    \begin{equation}
        \|\Phi_{\bfU_r}(\rho) - \Phi_{\bfU_r}(\sigma)\|_2 \geq \frac{\sqrt{d_A} \epsilon}{2d}.
    \end{equation}
    We now define the following parameters:
    \begin{equation}
        \label{eq:parameters-2}
        \epsilon^\prime = \frac{\sqrt{d_A} \epsilon}{2d}, \delta^\prime = \frac{1}{4C_2}, \tau = \frac{1}{2C_2}.
    \end{equation}
    In iteration $r$, $\mc{A}$ calls $\mathsf{HSCertify}(\Phi_{\bfU_r}(\sigma),\epsilon^\prime,\delta^\prime)$ on its copies of $\Phi_{\bfU_r}(\rho)$; using \Cref{lem:hscertify}, this can be done when $m^\prime = \bigo(\log(1/\delta^\prime)/\epsilon^{\prime 2}) = \bigo(d^2/d_A \epsilon^2)$. When $\|\rho-\sigma\|_1 \geq \epsilon$, we have
    \begin{equation}
        \Pr[\text{``Far''}] \geq (1-\delta^\prime) \cdot \Pr[\|\Phi_{\bfU_r}(\rho) - \Phi_{\bfU_r}(\sigma)\|_2 \geq \epsilon^\prime] \geq \left(1 - \frac{1}{4C_2}\right) \frac{1}{C_2} \geq \frac{3}{4C_2}.
    \end{equation}
    On the other hand, when $\rho = \sigma$, we have $\Phi_{\bfU_r}(\rho) = \Phi_{\bfU_r}(\sigma)$, and thus in this case, 
    \begin{equation}
        \Pr[\text{``Far''}] \leq \delta^\prime = \frac{1}{4C_2}.
    \end{equation}
    Thus, one can distinguish between the two cases by repeating this procedure and comparing the frequency of ``Far'' to $1/2C_2$. Using Hoeffding's inequality, this test succeeds with probability at least $1-\delta$ using $R = \bigo(C_2^2 \log(1/\delta)) = \bigo(\log(1/\delta))$ repetitions. The overall number of servers used is 
    \begin{equation}
        m = m^\prime \cdot R = \bigo\left(\frac{d^2}{d_A\epsilon^2} \cdot \log(1/\delta)\right),
    \end{equation}
    as desired.
\end{proof}

\subsection{Proof of \Cref{lem:mixed-state-compression}}

Our proof of \Cref{lem:mixed-state-compression} will use the following bounds on the moments of $\|\Phi_U(\rho) - \Phi_U(\sigma)\|_2$:

\begin{lemma}
\label{lem:expected-distance-improved}
    For quantum states $\rho,\sigma \in \mathbb{C}^{d \times d}$ with $d \geq d_A \geq 2$,
    \begin{equation}
        \mathbb{E}_{\bfU \sim U(d)} [\|\Phi_{\bfU}(\rho) - \Phi_{\bfU}(\sigma)\|_2^2] \geq \frac{d_A}{2d} \|\rho-\sigma\|_2^2.
    \end{equation}
\end{lemma}

We will also use the following bound on the fourth moment of the distance.

\begin{lemma}
    \label{lem:expected-distance-fourth-improved}
    For quantum states $\rho,\sigma \in \mathbb{C}^{d \times d}$ with $d$ sufficiently large, 
    \begin{equation}
        \mathbb{E}_{\bfU \sim U(d)} [\|\Phi_{\bfU}(\rho) - \Phi_{\bfU}(\sigma)\|_2^4] \leq C_1 \cdot \frac{d_A^2}{d^2} \|\rho-\sigma\|_2^4,
    \end{equation}
    where $C_1 > 1$ is an absolute constant.
\end{lemma}

Applying the Paley-Zygmund inequality to the random variable $\|\Phi_{\bfU}(\rho) - \Phi_{\bfU}(\sigma)\|_2^2$ and using the above \Cref{lem:expected-distance-improved,lem:expected-distance-fourth-improved}, we immediately obtain \Cref{lem:mixed-state-compression}. It remains only to prove \Cref{lem:expected-distance-improved,lem:expected-distance-fourth-improved}, which we do next.

\begin{proof}[Proof of \Cref{lem:expected-distance-improved}]
    Recall that we wish to bound $\mathbb{E}[\|\Phi_{\bfU}(\rho) - \Phi_{\bfU}(\sigma)\|_2^2] = \mathbb{E}[\tr((\Phi_{\bfU}(\rho) - \Phi_{\bfU}(\sigma))^2)]$. Let $\Delta = \rho - \sigma$. By linearity of $\Phi_{\bfU}$, we can rewrite the inner trace as 
    \begin{align}
        \tr(\Phi_{\bfU}(\Delta)^2) &= \tr(\swap_{A_1,A_2} \cdot \Phi_{\bfU}(\Delta)^{\otimes 2})
        \\&= \tr(\swap_{A_1,A_2} \otimes \mathbbm{1}_{B_1,B_2} \cdot ({\bfU}\Delta {\bfU}^\dag)^{\otimes 2}) \triangleq \tr(M_{\bfU}) \label{eq:second-moment-1},
    \end{align}
    where the second equality used $\Phi_{\bfU}(M) = \tr_B(UMU^\dag)$.
    By \Cref{lem:weingarten-calc,lem:weingarten-coeffs}, and noting that $\tr(\Delta) = 0$, we have
    \begin{align}
        \mathbb{E}[\tr(M_{\bfU})] &= \frac{\tr(\Delta^2)}{d^2-1}\tr(\swap_{A_1,A_2} \otimes \mathbbm{1}_{B_1,B_2} \cdot \swap_{1,2}) -\frac{\tr(\Delta^2)}{d(d^2-1)} \tr(\swap_{A_1,A_2} \otimes \mathbbm{1}_{B_1,B_2} \cdot \mathbbm{1}_{1,2})
        \\&= \frac{\tr(\Delta^2)}{d^2-1} \tr(\mathbbm{1}_{A_1,A_2} \otimes \swap_{B_1,B_2}) - \frac{\tr(\Delta^2)}{d(d^2-1)} \tr(\swap_{A_1,A_2} \otimes \mathbbm{1}_{B_1,B_2})
        \\&= \frac{\tr(\Delta^2)}{(d^2-1)} \left(d_A^2d_B - \frac{d_A d_B^2}{d}\right) = \frac{\tr(\Delta^2)}{(d^2-1)} \left(d_A^2d_B - d_B\right)
        \\&\geq \frac{\tr(\Delta^2)}{d^2} \cdot \frac{d_A^2 d_B}{2} 
        \\&= \frac{d_A}{2d} \cdot \|\rho - \sigma\|_2^2,
    \end{align}
    where the second equality used $\swap_{1,2} = \swap_{A_1,A_2} \otimes \swap_{B_1,B_2}$, the inequality holds for $d_A \geq 2$, and the final step used $d = d_Ad_B$ and $\tr(\Delta^2) = \tr((\rho-\sigma)^2) = \|\rho-\sigma\|_2^2$.
\end{proof}

\begin{proof}[Proof of \Cref{lem:expected-distance-fourth-improved}]

    Recalling the definition of $M_U$ from the proof of \Cref{lem:expected-distance-improved}, we have
    \begin{align}
        \mathbb{E}[\|\Phi_{\bfU}(\rho) - \Phi_{\bfU}(\sigma)\|_2^4] &= \mathbb{E}[\tr(M_{\bfU})^2] = \mathbb{E}[\tr(M_{\bfU} \otimes M_{\bfU})]
        \\&= \mathbb{E}[\tr(\swap_{A_1,A_2} \otimes \swap_{A_3,A_4} \otimes \mathbbm{1}_{B_1, \dots, B_4} \cdot {\bfU}^{\otimes 4} \Delta^{\otimes 4} {\bfU}^{\dag \otimes 4})].
    \end{align}
    Going forward, let $\gamma = \swap_{1,2} \otimes \swap_{3,4} \in \mc{S}_4$. We will now use \Cref{lem:weingarten-calc} with $k = 4$ to average over ${\bfU}$. As $\Delta$ is traceless, $\tr(P_\tau^\dag \Delta^{\otimes 4})$ in \Cref{lem:weingarten-calc} is only non-zero when $\tau$ has cycle type $(4)$ or $(2,2)$. Thus, we have
    \begin{align}
         \mathbb{E}[\tr(\gamma_A \otimes \mathbbm{1}_B \cdot {\bfU}^{\otimes 4} \Delta^{\otimes 4} {\bfU}^{\dag \otimes 4})] &= \sum_{\pi \in \mc{S}_4} \tr(\pi_{A,B} \cdot \gamma_A \otimes \mathbbm{1}_B) \sum_{\tau \in \mathrm{cyc}(2,2) \cup \mathrm{cyc}(4)} \wg(\pi^{-1} \tau) \tr(\tau^{-1} \Delta^{\otimes 4})
         \\&= \sum_{\pi \in \mc{S}_4} \tr(\pi_{A} \cdot \gamma_A) \tr(\pi_B) \sum_{\tau \in \mathrm{cyc}(2,2) \cup \mathrm{cyc}(4)} \wg(\pi^{-1} \tau) \tr(\tau^{-1} \Delta^{\otimes 4}).
    \end{align}
    We will also use the fact that $\tr(\Delta^4) = \|\Delta\|_4^4 \leq \|\Delta\|_2^4 = \tr(\Delta^2)^2$, where the inequality follows from monotonicity of Schatten norms. This implies that $\tr(\tau^{-1} \Delta^{\otimes 4}) \leq \tr(\Delta^2)^2$ for any choice of $\tau$. We now split the above summation into cases depending on $\pi$'s cycle type. Note that for any permutation $\pi$, we have $\tr(\pi) = d^{\# \pi}$, where $\# \pi$ is the number of cycles of $\pi$.
    
    First consider the case when $\pi \in \mathrm{cyc}(4)$. Note that $\tr(\pi_B) = d_B$, and that, by \Cref{lem:weingarten-coeffs}, $\wg(\pi^{-1}\tau) \leq d^{-4}$ for any choice of $\tau$. Further, as $\pi$ and $\gamma$ have different cycle types, $\pi \cdot \gamma \neq \mathrm{id}$, implying that $\tr(\pi_A \cdot \gamma_A) \leq d_A^3$. As there is only a constant number of choices for $\pi, \tau$, we have
    \begin{equation}
        \sum_{\pi \in \mathrm{cyc}(4)} \tr(\pi_{A} \cdot \gamma_A) \tr(\pi_B) \sum_{\tau \in \mathrm{cyc}(2,2) \cup \mathrm{cyc}(4)} \wg(\pi^{-1} \tau) \tr(\tau^{-1} \Delta^{\otimes 4}) \leq \bigo \left(\frac{d_A^3 d_B}{d^4}\right) \cdot \tr(\Delta^2)^2 = \bigo \left(\frac{d_A^2}{d^3}\right) \cdot \tr(\Delta^2)^2. 
    \end{equation}

   Next, assume $\pi \in \mathrm{cyc}(3,1)$. As $\pi$ has two cycles, $\tr(\pi_B) = d_B^2$. Since $\pi,\tau$ have different cycle types, by \Cref{lem:weingarten-coeffs}, $\wg(\pi^{-1}\tau) \leq \bigo(d^{-6})$ for any choice of $\tau$. Again, $\pi,\gamma$ have different cycle types, and so $\tr(\pi_A \cdot \gamma_A) \leq d_A^3$. Altogether,

   \begin{equation}
        \sum_{\pi \in \mathrm{cyc}(3,1)} \tr(\pi_{A} \cdot \gamma_A) \tr(\pi_B) \sum_{\tau \in \mathrm{cyc}(2,2) \cup \mathrm{cyc}(4)} \wg(\pi^{-1} \tau) \tr(\tau^{-1} \Delta^{\otimes 4}) \leq \bigo \left(\frac{d_A^3 d_B^2}{d^6}\right) \cdot \tr(\Delta^2)^2 = \bigo \left(\frac{d_A}{d^4}\right) \cdot \tr(\Delta^2)^2. 
    \end{equation}

    In the next case, let $\pi \in \mathrm{cyc}(2,2)$. As $\pi$ has two cycles, $\tr(\pi_B) = d_B^2$, and, by \Cref{lem:weingarten-coeffs}, $\wg(\pi^{-1}\tau) \leq d^{-4}$ for any choice of $\tau$. Finally, as both $\pi$ and $\gamma$ have the same cycle type, $\tr(\pi_A \cdot \gamma_A)$ is maximized by $\pi = \gamma^{-1}$, implying $\tr(\pi_A \cdot \gamma_A) \leq d_A^4$. Then, 
    \begin{equation}
         \sum_{\pi \in \mathrm{cyc}(2,2)} \tr(\pi_{A} \cdot \gamma_A) \tr(\pi_B) \sum_{\tau \in \mathrm{cyc}(2,2) \cup \mathrm{cyc}(4)} \wg(\pi^{-1} \tau) \tr(\tau^{-1} \Delta^{\otimes 4}) \leq \bigo \left(\frac{d_A^4 d_B^2}{d^4}\right) \cdot \tr(\Delta^2)^2 = \bigo \left(\frac{d_A^2}{d^2}\right) \cdot \tr(\Delta^2)^2. 
    \end{equation}

   Next, consider the case when $\pi \in \mathrm{cyc}(2,1,1)$. As $\pi$ has 3 cycles, $\tr(\pi_B) = d_B^3$. Using \Cref{lem:weingarten-coeffs}, and as we sum over $\pi,\tau$ of different cycle types, we have $\wg(\pi^{-1}\tau) \leq \bigo(d^{-6})$. Again, as $\pi$ and $\gamma$ have different types, $\tr(\pi_A \gamma_A) \leq d_A^3$, implying

   \begin{equation}
         \sum_{\pi \in \mathrm{cyc}(2,1,1)} \tr(\pi_{A} \cdot \gamma_A) \tr(\pi_B) \sum_{\tau \in \mathrm{cyc}(2,2) \cup \mathrm{cyc}(4)} \wg(\pi^{-1} \tau) \tr(\tau^{-1} \Delta^{\otimes 4}) \leq \bigo \left(\frac{d_A^3 d_B^3}{d^6}\right) \cdot \tr(\Delta^2)^2 = \bigo \left(\frac{1}{d^3}\right) \cdot \tr(\Delta^2)^2. 
    \end{equation}

    In the final case, $\pi \in \mathrm{cyc}(1,1,1,1)$, i.e., $\pi = \mathrm{id}$. Note that $\tr(\pi_B) = d_B^4$, and again, as $\pi \neq \tau$, $\wg(\pi^{-1}\tau) \leq \bigo(d^{-6})$ for any choice of $\tau$. Finally, $\tr(\pi_A \gamma_A) = \tr(\gamma_A) = d_A^2$. This yields the bound
    \begin{equation}
        \sum_{\pi \in \mathrm{cyc}(1,1,1,1)} \tr(\pi_{A} \cdot \gamma_A) \tr(\pi_B) \sum_{\tau \in \mathrm{cyc}(2,2) \cup \mathrm{cyc}(4)} \wg(\pi^{-1} \tau) \tr(\tau^{-1} \Delta^{\otimes 4}) \leq \bigo \left(\frac{d_A^2 d_B^4}{d^6}\right) \cdot \tr(\Delta^2)^2 = \bigo \left(\frac{d_B^2}{d^4}\right) \cdot \tr(\Delta^2)^2.
    \end{equation}

    It is easy to verify that the largest contribution is due to the third case, where $\pi$ has cycle type $(2,2)$. Thus, we obtain
    \begin{equation}
        \mathbb{E}[\|\Phi_U(\rho) - \Phi_U(\sigma)\|_2^4] \leq \bigo\left(\frac{d_A^2}{d^2}\right) \cdot \|\rho-\sigma\|_2^4,
    \end{equation}
    concluding the proof.

\end{proof}

\section{Lower bound framework}\label{s:lower-bound-framework}

In this section, we present the main technical framework used in our lower bound proofs. Recall, as discussed in Section~\ref{ss:technical-overview}, that our techniques generalize those introduced for the analysis of distributed classical inference in Ref.~\cite{acharya2020inferenceinformationconstraints}. Now, to prove lower bounds for certification in our setting, we will consider its hardest instance, i.e., mixedness testing. Specifically, we will prove lower bounds for distinguishing between the maximally mixed state $\mmstate$ and an ensemble of states that are $\eps$-far from $\mmstate$. Note that any algorithm for mixedness testing must also succeed at this latter task with high probability. Let us first formally define such ``almost-$\eps$ perturbation ensembles'', borrowing terminology from \cite{acharya2020inferenceinformationconstraints}.

\begin{definition}
    \label{def:almost-eps-perturbations}
    An ensemble $D$ of quantum states is an almost-$\eps$ perturbation of a state $\sigma$ if $\mathrm{Pr}_{\rho \sim D}[\|\rho - \sigma\|_1 \geq \eps] \geq \frac12$. Let the set of all almost-$\eps$ perturbations of $\sigma$ be $\mc{D}_{\eps}(\sigma)$.
\end{definition}

To prove lower bounds for distinguishing between a fixed state $\sigma$ and a state drawn from some almost-$\eps$ perturbed ensemble $D$, one typically shows that the states $\sigma^{\otimes m}$ and $\mathbb{E}_{\rho \sim D}[\rho^{\otimes m}]$ are statistically indistinguishable unless the number of copies $m$ is large enough. However, in our distributed setting, we can show stronger lower bounds than permitted by such arguments. 

In particular, consider the setting where each distributed node receives a copy of some unknown $d$-dimensional state $\rho$, and can send a $d_q$-dimensional qudit to the central node, with $d_q \leq d$. Then, the most general action of the distributed node $N_i$ can be modeled as a quantum channel $\Phi_i : \mathbb{C}^{d \times d} \mapsto \mathbb{C}^{d_q \times d_q}$. The central node then receives states $\Phi_1(\rho), \dots, \Phi_m(\rho)$, and must determine whether $\rho = \sigma$ or $\rho$ was drawn from $D$. Thus, we can now argue that $m$ must be large enough for $\bigotimes_{i = 1}^m \Phi_i(\sigma)$ and $\mathbb{E}_{\rho \sim D} \left[\bigotimes_{i = 1}^m \Phi_i(\rho)\right]$ to be statistically distinguishable. Moreover, when we consider the private-coin setting, we can show even stronger lower bounds. In this weaker setting, one can imagine that the ensemble $D$ is chosen adversarially depending on the choice of each $\Phi_i$. We formalize these arguments in the following lemma: 

\begin{lemma}
\label{lem:copy-complexity-and-quantum-chi2}
Let $d \geq d_q \geq 2$. Then, for any public-coin protocol to succeed at $\eps$-certifying $d$-dimensional states using only $d_q$-dimensional quantum messages, the number of distributed nodes $m$ must be large enough such that
\begin{equation}
   \min_{D \in \mc{D}_{\eps}(\mathbbm{1}/d)} \hspace{.5em} \max_{\Phi_1, \dots, \Phi_m} \qdchi\left(\mathbb{E}_{\rho \sim D}\left[\bigotimes_{i = 1}^m \Phi_i(\rho)\right] \Bigg \|\bigotimes_{i = 1}^m \Phi_i\left(\mmstate\right)\right) \geq \frac{1}{16}.
\end{equation}
On the other hand, for any private-coin protocol to succeed in the same setting, we must have
\begin{equation}
     \max_{\Phi_1, \dots, \Phi_m} \hspace{.5em} \min_{D \in \mc{D}_{\eps}(\mathbbm{1}/d)} \qdchi\left(\mathbb{E}_{\rho \sim D}\left[\bigotimes_{i = 1}^m \Phi_i(\rho)\right] \Bigg \|\bigotimes_{i = 1}^m \Phi_i\left(\mmstate\right)\right) \geq \frac{1}{16}.
\end{equation}
\end{lemma}

The above lemma follows from a generalization of standard arguments to the quantum setting (see e.g. \cite{acharya2020inferenceinformationconstraints,liu2024role}). Nevertheless, for the sake of completeness, we include the explicit proof in \Cref{s:maxmin-minmax-divergence}. 

\Cref{lem:copy-complexity-and-quantum-chi2} essentially allows us to view the point-versus-mixture distinguishing task as a two-player game, where one player selects the channels $\Phi_1, \dots, \Phi_m$, and an adversary picks the mixture of alternatives $D \in \mc{D}_\eps(\mathbbm{1}/d)$. The former attempts to maximize the quantum $\chi^2$-divergence while the latter attempts to minimize this quantity. The key difference between the public and private coin settings is the order in which the players make their moves. The adversary's ability to go second in the private-coin setting allows them to pick a hard mixture \emph{depending} on the first player's chosen channels, allowing us to show stronger lower bounds.

From the above, it's clear that to lower bound the number of nodes $m$ necessary, it suffices to \textit{upper bound} the quantum $\chi^2$ divergences appearing in \Cref{lem:copy-complexity-and-quantum-chi2}. Now, to prove such upper bounds, we will use the recently developed quantum analogue of the Ingster--Suslina method from \cite[Lemma 4.3]{odonnell2025instanceoptimalquantumstatecertification}, stated below:
\begin{lemma}
\label{lem:quantum-ingster-suslina}
    Let $\bftheta$ be a random variable parameterizing $m$ quantum states $\sigma_{1, \bftheta}, \dots, \sigma_{m, \bftheta} \in \mathbb{C}^{d \times d}$. Let $\sigma_{\bftheta}^{(m)} \triangleq \bigotimes_{i = 1}^m \sigma_{i, \bftheta}$. Let $\sigma^{(m)} \triangleq \sigma_1 \otimes \dots \otimes \sigma_m$, where each $\sigma_i \in \mathbb{C}^{d \times d}$. Then
    \begin{equation}
        \qdchi(\mathbb{E}_{\bftheta}[\sigma_{\bftheta}^{(m)}] \| \sigma^{(m)}) = \mathbb{E}_{\bftheta,\bftheta^\prime} \left[
            \prod_{i = 1}^m (1 + Z_i(\bftheta,\bftheta^\prime))
        \right] - 1 \leq \mathbb{E}_{\bftheta,\bftheta^\prime} \left[
            \prod_{i = 1}^m \exp(Z_i(\bftheta,\bftheta^\prime))\right],
    \end{equation}
    where $\bftheta,\bftheta^\prime$ are i.i.d and
    \begin{equation}
        Z_i(\bftheta,\bftheta^\prime) = \tr\left(
            \sigma_i^{-1}(\sigma_{i, \bftheta} - \sigma_i)(\sigma_{i, \bftheta^\prime} - \sigma_i)
        \right).
    \end{equation}
\end{lemma}
\begin{remark}
    We remark that \Cref{lem:quantum-ingster-suslina} is stated slightly more generally than the original \cite[Lemma 4.3]{odonnell2025instanceoptimalquantumstatecertification}, which assumed $\sigma_{1, \bftheta} =  \dots = \sigma_{m, \bftheta}$ and $\sigma_1 = \dots = \sigma_m$. We state this general version without proof, as it would be essentially identical to that of the original lemma. Our generalization is necessary in the distributed setting, where the action of each node may not be identical, resulting in different states.
\end{remark} 

Finally, let us state the almost-$\eps$ perturbation family we will consider in our lower bounds. This instance was first considered in \cite{liu2024role} and has since been used to prove lower bounds in various other resource-constrained settings (see e.g. \cite{liu2024quantum,aliakbarpour2025adversarially}).

\begin{definition}
\label{def:hard-instance-mic}
    Let $\{V_i\}_{i \in [d^2]}$ form an orthonormal basis for $\mathbb{C}^{d \times d}$ w.r.t the Hilbert--Schmidt inner product, with $V_{d^2} = \mathbbm{1}/\sqrt{d}$. For some integer $\frac{d^2}{2} \leq \ell \leq d^2-1$, and $z \in \{-1,+1\}^{\ell}$, define the parameterized perturbation
    \begin{equation}
        \Delta_z \triangleq \frac{c\eps}{\sqrt{d}} \cdot \frac{1}{\sqrt{\ell}} \sum_{i = 1}^\ell z_i V_i, \quad \text{and } \quad \bar{\Delta}_z \triangleq \Delta_z \cdot \min\left\{1, \frac{1}{d\|\Delta_z\|_{\infty}}\right\}.
    \end{equation}
    Let $\rho_z \triangleq \frac{\mathbbm{1}}{d} + \bar{\Delta}_z$. It will be convenient to define the matrix $\mc{V} = [\mathrm{vec}(V_1), \dots, \mathrm{vec}(V_\ell)] \in \mathbb{C}^{d^2 \times \ell}$.
\end{definition}

It was shown in \cite[Corollary 4.4]{liu2024role} that for any choice of $\{V_i\}$ and a uniformly random $\bfz \sim \{-1,+1\}^\ell$, the parameterized state $\rho_{\bfz}$ is $\eps$-far from $\mmstate$ with high probability, i.e., the above ensemble is a valid almost $\eps$-perturbation, whenever $\eps$ is smaller than an absolute constant $C$.

\section{Lower bounds}\label{s:lower-bounds}

In this section, we will prove our lower bounds for distributed state certification with only quantum communication, i.e., the $(n_c = 0, n_q, R, E = 0)$-setting, where $R \in \{\mathsf{public},\mathsf{private}\}$. As before, we will consider the more general case where each distributed node can send a single $d_q$-dimensional qudit. First, let us state our public-coin lower bound in this case:
\begin{theorem}
    \label{thm:public-lower-bound}
    Let $d \geq 6, 1 \leq d_q \leq d, 0 < \eps < C$, where $C$ is an absolute constant. Consider the public-coin distributed setting with $m$ distributed nodes, each capable of sending a qudit of dimension $d_q$. Suppose the action of each distributed node can be modeled by a mixedness-preserving channel. Then, to $\eps$-certify a $d$-dimensional state, the number of distributed nodes $m$ must be at least $\Omega\left(\frac{d^2}{d_q \eps^2}\right)$.
\end{theorem}

Now, let us state our private-coin lower bound:
\begin{theorem}
    \label{thm:private-lower-bound}
    Let $d \geq 6, 1 \leq d_q \leq d, 0 < \eps < C,$ where $C$ is an absolute constant. Consider the private-coin distributed setting with $m$ distributed nodes, each capable of sending a qudit of dimension $d_q$. Suppose the action of each distributed node can be modeled by a mixedness-preserving channel. Then, to $\eps$-certify a $d$-dimensional state, the number of distributed nodes $m$ must be at least $\Omega\left(\frac{d^3}{d_q^{2} \eps^2}\right)$.
\end{theorem}

To prove these lower bounds, our central technical result is the following lemma that relates the quantum $\chi^2$-divergence to appropriate norms of mixedness-preserving channels:

\begin{lemma} 
\label{lem:qdchi-upper-bound}
    Let $d \geq 6,0 < \eps < C$ and $1 \leq d_q \leq d$. Let $\rho_z, \mc{V}$ be as defined in \Cref{def:hard-instance-mic}. For $i \in [m],$ let $\Phi_i : \mathbb{C}^{d \times d} \mapsto \mathbb{C}^{d_q \times d_q}$ be a mixedness-preserving quantum channel. We define the $\mathbb{C}^{d \times d}$ operator $T({\Phi_1, \dots, \Phi_m}) \triangleq \frac{1}{m} \sum_{i = 1}^m M_{\Phi_i}^\dag M_{\Phi_i}$, where $M_{\Phi_i}$ denotes the Liouville matrix representation of $\Phi$. Then, $\qdchi\left(\mathbb{E}_{\bfz}[\bigotimes_{i = 1}^m \Phi_i(\rho_{\bfz})] \| \bigotimes_{i = 1}^m \Phi_i(\mathbbm{1}/d)\right) \leq \frac{1}{16}$ unless
    \begin{equation}
        m \geq \Omega\left(\frac{d\ell}{d_q \eps^2} \cdot \frac{1}{\|\mc{V}^\dag T_{\Phi_1, \dots, \Phi_m} \mc{V}\|_2}\right).
    \end{equation}
\end{lemma}

Additionally, we will need the following bounds on the norms of the Liouville matrix:

\begin{lemma}
\label{lem:unital-channel-frobenius-norm-upper-bound}
    Let $d \geq d_q \geq 1$. For all mixedness-preserving quantum channels $\Phi: \mathbb{C}^{d \times d} \mapsto \mathbb{C}^{d_q \times d_q}$, we have
    \begin{equation}
        \|M_{\Phi}\|_2 \leq \sqrt{dd_q}
    \end{equation}
\end{lemma}

\begin{lemma}
\label{lem:unital-channel-operator-norm-upper-bound}
    Let $d \geq d_q \geq 1$. For all mixedness-preserving quantum channels $\Phi: \mathbb{C}^{d \times d} \mapsto \mathbb{C}^{d_q \times d_q}$, we have
    \begin{equation}
        \|M_{\Phi}\|_\infty \leq \sqrt{\frac{d}{d_q}}.
    \end{equation}
\end{lemma}

We defer the proof of \Cref{lem:qdchi-upper-bound} to \Cref{ss:qdchi-upper-bound} and those of \Cref{lem:unital-channel-frobenius-norm-upper-bound,lem:unital-channel-operator-norm-upper-bound} to \Cref{ss:norm-bounds}. Now, given the above lemmas, we can prove our main lower bound theorems. Let us start with the public coin setting:

\begin{proof}[Proof of \Cref{thm:public-lower-bound}]
    Let the action of each distributed node $N_i$ be some mixedness-preserving quantum channel $\Phi_i : \mathbb{C}^{d \times d} \mapsto \mathbb{C}^{d_q \times d_q}$ applied to their copy of the unknown state $\rho$. By \Cref{lem:copy-complexity-and-quantum-chi2}, for all choices of $\mc{V}$, the number of copies $m$ to succeed at state certification must be large enough such that 
    \begin{equation}
    \max_{\Phi_1, \dots, \Phi_m} \qdchi(\mathbb{E}_{\bfz}[\bigotimes_{i = 1}^m \Phi_i(\rho_{\bfz})] \| \bigotimes_{i = 1}^m \Phi_i(\mathbbm{1}/d)]) \geq \frac{1}{16}.
    \end{equation}
    However, by \Cref{lem:qdchi-upper-bound}, for this to hold, we must have
    \begin{equation}
        m \geq \Omega\left(\frac{d\ell}{d_q \eps^2} \cdot \min_{\Phi_1, \dots, \Phi_m} \frac{1}{\|\mc{V}^\dag T({\Phi_1, \dots, \Phi_m}) \mc{V}\|_2}\right), \label{eq:public-lower-bound-1}
    \end{equation}
    Now note that, for any valid choice of $V_1, \dots, V_{d^2-1}$, $\mc{V}$ is an isometry. Consequently, for all choices of $\Phi_1, \dots, \Phi_m$, we have
    \begin{equation}
        \|\mc{V}^\dag T({\Phi_1, \dots, \Phi_m}) \mc{V}\|_2 \leq \|T({\Phi_1, \dots, \Phi_m})\|_2 \leq \frac{1}{m} \sum_{i = 1}^m\|M_{\Phi_i}^\dag M_{\Phi_i}\|_2 \leq \frac{1}{m} \sum_{i = 1}^m\|M_{\Phi_i}^\dag\|_2 \|M_{\Phi_i}\|_\infty \leq d,
    \end{equation}
    where we used the triangle inequality in the second step and applied \Cref{lem:unital-channel-frobenius-norm-upper-bound,lem:unital-channel-operator-norm-upper-bound} to each channel $\Phi_i$ in the last inequality. Now, substituting this norm bound and $\ell = d^2-1$ back into \Cref{eq:public-lower-bound-1}, we get
    \begin{equation}
         m \geq \Omega\left(\frac{d^2}{d_q \eps^2}\right),
    \end{equation}
    as desired.
\end{proof}

Next, we will prove our private-coin lower bound:

\begin{proof}[Proof of \Cref{thm:private-lower-bound}]
    In the private coin case, using \Cref{lem:copy-complexity-and-quantum-chi2}, we instead need the max-min divergence to be larger than $\frac12$. In other words, we can choose the perturbation operator $\mc{V}$ adversarially to minimize the quantum $\chi^2$-divergence. From \Cref{lem:qdchi-upper-bound}, this corresponds to minimizing $\|\mc{V}^\dag T(\Phi_1, \dots, \Phi_m) \mc{V}\|_2$ for any fixed choice of channels $\Phi_1, \dots, \Phi_m$. For a fixed parameter $\ell$ and a fixed choice of channels, this can be done by choosing the columns of $\mc{V}$ to be the $\ell$ eigenvectors of $T(\Phi_1,\dots,\Phi_m)$ corresponding to its $\ell$ smallest eigenvalues.
    
    Let these eigenvalues be $\lambda_1, \dots, \lambda_\ell$ and the eigenvectors (with unit norm) be $\mathrm{vec}(V_1), \dots, \mathrm{vec}(V_\ell)$, where we reserve $V_{d^2} = \frac{\mathbbm{1}}{\sqrt{d}}$. Note that the operators $V_1, \dots, V_\ell$ are traceless, normalized, and pairwise orthogonal with respect to the Hilbert--Schmidt inner product. This can be seen by noting that $T(\Phi_1,\dots,\Phi_m)$ is Hermitian and so has a spectral decomposition into a basis of orthonormal eigenvectors; further, by assumption, the identity operator is an eigenvector of each channel $\Phi_i$, and thus $\mathrm{vec}(\mathbbm{1}_d)$ is an eigenvector of $T(\Phi_1,\dots,\Phi_m)$. 
    Consequently, we can write
    \begin{equation}
        \|\mc{V}^\dag T(\Phi_1, \dots, \Phi_m) \mc{V}\|_2 = \sqrt{\sum_{i = 1}^\ell \lambda_\ell^2}.
    \end{equation}
    Now, note that each of the $\ell$ smallest eigenvalues $\lambda_1, \dots, \lambda_\ell$ is at most the average of the remaining eigenvalues $\lambda_{\ell+1},\dots, \lambda_{d^2-1}$. Thus,
    \begin{align}
         \|\mc{V} T(\Phi_1, \dots, \Phi_m) \mc{V}\|_2  &\leq \sqrt{\ell} \cdot \frac{\sum_{i = \ell+1}^{d^2-1} \lambda_i}{d^2-\ell-1} \leq \sqrt{\ell} \cdot \frac{\tr(T(\Phi_1, \dots, \Phi_m))}{d^2-\ell-1} 
         \\&= \frac{\sqrt{\ell}}{d^2-\ell-1} \cdot \frac{\sum_{i = 1}^m \tr(M^\dag_{\Phi_i} M_{\Phi_i})}{m} = \frac{\sqrt{\ell}}{d^2-\ell-1} \cdot \frac{\sum_{i = 1}^m \|M_{\Phi_i}\|_2^2}{m}
         \\&\leq \frac{\sqrt{\ell}}{d^2-\ell-1} \cdot dd_q.
    \end{align}
    where we used the fact that $T(\Phi_1, \dots, \Phi_m)$ is positive semidefinite in the second inequality and \Cref{lem:unital-channel-frobenius-norm-upper-bound} in the last step. Substituting this norm bound with $\ell = d^2/2$ back into \Cref{lem:qdchi-upper-bound}, we get
    \begin{equation}
        m \geq \Omega\left(\frac{d^3}{d_q^{2} \eps^2}\right),
    \end{equation}
    concluding the proof.
\end{proof}

\subsection{Proof of \Cref{lem:qdchi-upper-bound}}
\label{ss:qdchi-upper-bound}

\begin{proof}
     Using \Cref{lem:quantum-ingster-suslina}, we have
    \begin{align}
         1 + \qdchi\left(\mathbb{E}_{\bfz} \left[ \bigotimes_{i = 1}^m \Phi_i(\rho_{\bfz})\right] \bigg\| \bigotimes_{i = 1}^m \Phi_i\left(\mmstate\right)\right) &\leq \mathbb{E}_{\bfz,\bfz^\prime} \left[\prod_{i = 1}^m \exp(Z_{\Phi_i}(\bfz,\bfz^\prime))\right]
         \\&=  \mathbb{E}_{\bfz,\bfz^\prime}  \exp( \sum_{i = 1}^m Z_{\Phi_i}(\bfz,\bfz^\prime)),
    \end{align}
    where, for any channel $\Phi$, we define $Z_\Phi(z,z^\prime) \triangleq \tr(\Phi(\mathbbm{1}/d)^{-1}\Phi(\bar{\Delta}_z)\Phi(\bar{\Delta}_{z^\prime}))$. Now,
    \begin{align}
         \Phi(\bar{\Delta}_z) &= \frac{c\eps N_z}{\sqrt{d\ell}} \cdot \Phi\left(\sum_{i = 1}^\ell z_i V_i\right)
         \\&= \frac{c\eps N_z}{\sqrt{d\ell}} \cdot M_\Phi \mathrm{vec}\left(\sum_{i = 1}^\ell V_i z_i\right)
         \\&= \frac{c\eps N_z}{\sqrt{d\ell}} \cdot M_\Phi \mc{V} z.
    \end{align}
    We can now compute $Z_{\Phi}(z,z^\prime)$:
    \begin{align}
     Z_{\Phi}(z,z^\prime) &= \tr\left(\Phi\left(\frac{\mathbbm{1}}{d}\right)^{-1} \Phi(\Delta_z) \Phi(\Delta_{z^\prime}) \right)
     \\&= \tr\left(d_q \cdot \mathbbm{1} \Phi(\Delta_z) \Phi(\Delta_{z^\prime}) \right)
     \\&= \frac{d_q c^2 \eps^2 N_z N_{z^\prime}}{d \ell} (M_\Phi \mc{V} z)^\dag (M_\Phi \mc{V} z^\prime)
     \\&= \frac{d_q c^2\eps^2 N_z N_{z^\prime}}{d\ell}z^\top \mc{V}^\dag M_{\Phi}^\dag M_{\Phi} \mc{V} z^\prime,
    \end{align}
    where we used the assumption that $\Phi$ is mixedness preserving in the second equality.
    Moreover, we can write
    \begin{equation}
        \sum_{i = 1}^m Z_{\Phi_i}(z,z^\prime) = 
        \frac{d_q c^2\eps^2 N_z N_{z^\prime}}{d\ell} \cdot m \cdot \frac{\sum_{i = 1}^m \bfz^\top \mc{V}^\dag M_{\Phi_i}^\dag M_{\Phi_i} \mc{V} z^\prime}{m} 
        = \frac{md_q c^2\eps^2 N_z N_{z^\prime}}{d\ell} \cdot \bfz^\top \mc{V}^\dag T(\Phi_1, \dots, \Phi_m) \mc{V} z^\prime,
    \end{equation}
    where recall $T(\Phi_1, \dots, \Phi_m) = \frac1m \cdot \sum_{i = 1}^m M_{\Phi}^\dag M_{\Phi}$, and in the first equality above, we used the fact that $\Phi$ is mixedness preserving. Now, we have
    \begin{align}
         1 + \qdchi\left(\mathbb{E}_z \left[ \bigotimes_{i = 1}^m \Phi_i(\rho_{\bfz})\right] \bigg\| \bigotimes_{i = 1}^m \Phi_i\left(\mmstate\right)\right) &\leq  \mathbb{E}_{\bfz,\bfz^\prime}  \exp( \sum_{i = 1}^m Z_{\Phi_i}(\bfz,\bfz^\prime))
         \\&= \mathbb{E}_{\bfz,\bfz^\prime}  \exp\left( \frac{md_q c^2\eps^2 N_z N_{z^\prime}}{d\ell} \cdot \bfz^\top \mc{V}^\dag T(\Phi_1, \dots \Phi_m) \mc{V} \bfz\right)
         \\&\leq  \mathbb{E}_{z,z^\prime}  \exp\left( \frac{md_q c^2\eps^2}{d\ell} \cdot \bfz^\top \mc{V}^\dag T(\Phi_1, \dots \Phi_m) \mc{V} \bfz\right) + \frac{4}{e^d},
    \end{align}
    where the last inequality follows from an argument identical to the proof of \cite[Lemma B.8]{liu2024quantum}, using the fact that $N_{\bfz}$ is $1$ except with exponentially small probability. Thus, using \Cref{lem:mgf-quadratic-form}, we obtain
    \begin{align}
        1 + \qdchi\left(\mathbb{E}_{\bfz} \left[ \bigotimes_{i = 1}^m \Phi_i(\rho_z)\right] \bigg\| \bigotimes_{i = 1}^m \Phi_i(\mathbbm{1}/d)\right) &\leq  
         \exp\left(\frac{Cm^2 d_q^{2} \eps^4}{d^2 \ell^2}\|\mc{V}^\dag T(\Phi_1, \dots \Phi_m) \mc{V}\|_2^2\right) + \frac{4}{e^d}
         , \label{eq:main-lower-bound-1}
    \end{align}
    where the inequality holds assuming $m \leq \bigo\left(\frac{d\ell}{d_q \eps^2 \|\mc{V}^\dag T(\Phi_1, \dots \Phi_m) \mc{V}\|_{\infty}}\right)$. For $d \geq 6$, we have $4\exp(-d) < \frac{1}{100}$, and so, for the quantum $\chi^2$-divergence to be at least $\frac{1}{16}$, we must either have
    \begin{equation}
         m \geq \Omega\left(\frac{d\ell}{d_q \eps^2}\cdot\frac1{\|\mc{V}^\dag T(\Phi_1, \dots \Phi_m) \mc{V}\|_2}\right)
    \end{equation} 
    to make the right hand side of \Cref{eq:main-lower-bound-1} sufficiently large, or, have
    \begin{equation}
        m \geq \Omega\left(\frac{d\ell}{d_q \eps^2}\cdot \frac1{\|\mc{V}^\dag T(\Phi_1, \dots \Phi_m) \mc{V}\|_\infty}\right) \geq \Omega\left(\frac{d\ell}{d_q \eps^2}\cdot\frac1{\|\mc{V}^\dag T(\Phi_1, \dots \Phi_m) \mc{V}\|_2}\right)
    \end{equation}
    to falsify the assumption allowing the use of \Cref{lem:mgf-quadratic-form} above. Thus, in either case, our desired lower bound holds. 
\end{proof}

\subsection{Channel norm upper bounds}
\label{ss:norm-bounds}

We will first prove the Schatten $2$-norm bound:

\begin{proof}[Proof of \Cref{lem:unital-channel-frobenius-norm-upper-bound}]

As discussed in \Cref{s:preliminaries}, the Choi operator of a channel and its Liouville matrix representation are related by rearranging the entries; consequently, as the Schatten $2$-norm is also the entrywise $\ell_2$-norm, both matrices have the same Schatten $2$-norm. We will thus aim to show that $\tr(J(\Phi)^2) \leq dd_q$. Recall that $J(\Phi) \in \mathbb{C}^{d \times d} \otimes \mathbb{C}^{d_q \times d_q}$. Denoting these as registers $A,B$, recall that $\tr_B(J(\Phi)) = \mathbbm{1}_d$ and $\tr(J(\Phi)) = d$.

Next, it will be convenient to write
\begin{equation}
    \tr(J(\Phi)^2) \leq \|J(\Phi)\|_\infty \|J(\Phi)\|_1 = \|J(\Phi)\|_\infty \cdot \tr(J(\Phi)) = d \cdot \|J(\Phi)\|_\infty,
\end{equation}
where we used the matrix version of H\"older's inequality in the first step, and the fact that $J(\Phi) \succeq 0$ and so $\|J(\Phi)\|_1 = \tr(J(\Phi))$ in the next step. Thus, it suffices to show that $\|J(\Phi)\|_\infty \leq d_q$.

Let $\ket{\psi}$ be an arbitrary unit eigenvector of $J(\Phi)$ with eigenvalue $\lambda$. As $\psi$ is an arbitrary eigenvector, it is enough to show that $\lambda \leq d_q$. Now, we can always write a Schmidt decomposition of $\ket{\psi}$, i.e., 
\begin{equation}
    \ket{\psi} = \sum_{i = 1}^{d_q} \sqrt{p_i} \ket{a_i} \ket{b_i},
\end{equation}
for some $p_i \geq 0$ and $\sum_i p_i = 1$, and with $\{\ket{a_i}\}_{i \in [d_q]}$ an orthonormal set in $\mathbb{C}^d$ and $\{\ket{b_i}\}_{i \in [d_q]}$ an orthonormal basis of $\mathbb{C}^{d_q}$. It will be convenient to define the matrix $M \in \mathbb{C}^{d_q \times d_q}$ with entries 
\begin{equation}
    M_{i,j} \triangleq \bra{a_i} \otimes \bra{b_i} J(\Phi) \ket{a_j} \otimes \ket{b_j}, \quad \forall i,j \in [d_q].
\end{equation}

First, we relate this eigenvalue $\lambda$ to $\tr(M)$. We start by defining the unit vector $\ket{v} = \sum_{i = 1}^{d_q} \sqrt{p_i} \ket{i}$, where $\{\ket{i}\}$ is the basis $M$ is written in. Now, we can write
\begin{align}
    \lambda &= \bra{\psi} J(\Phi) \ket{\psi}
    \\&= \sum_{i,j = 1}^{d_q} \sqrt{p_i p_j} \bra{a_i} \otimes \bra{b_i} J(\Phi) \ket{a_j} \otimes \ket{b_j}
    \\&= \sum_{i,j = 1}^{d_q} \sqrt{p_i p_j} M_{i,j}
    \\&= \bra{v} M \ket{v}
    \\&\leq \|M\|_\infty
    \\&\leq \tr(M),
\end{align}
where in the last line we use that $M$ is a principal submatrix of $J(\Phi)$ and is thus positive semidefinite. 

It remains to show $\tr(M) \leq d_q$. To upper bound the trace, we will upper bound each diagonal entry of $M$. For any $i \in [d_q]$, we have

\begin{align}
    M_{i,i} &= \bra{a_i} \otimes \bra{b_i} J(\Phi) \ket{a_i} \otimes \ket{b_i}
    \\&\leq \sum_{j \in [d_q]} \bra{a_i} \otimes \bra{b_j} J(\Phi) \ket{a_i} \otimes \ket{b_j}
    \\&= \bra{a_i} \tr_B(J(\Phi)) \ket{a_i}
    \\&= \bra{a_i} \mathbbm{1}_d \ket{a_i} = 1.
\end{align}
Thus, we have
\begin{equation}
    \tr(M) = \sum_{i = 1}^{d_q} M_{i,i} \leq d_q,
\end{equation}
as desired.
\end{proof}

Before moving on to the proof of \Cref{lem:unital-channel-operator-norm-upper-bound}, let us first state an important inequality:

\begin{lemma}[Kadison-Schwarz inequality; see e.g., {\cite[Proposition 3.3]{paulsen2002completely}}]
\label{lem:kadison-schwarz}
Let $\Psi:\mathcal{A}\rightarrow \mathcal{B}$ be a unital completely positive map, then for all $Y\in \mathcal{A}$ one has
\begin{equation}
\Psi(Y^\dagger Y) - \Psi(Y)^\dagger \Psi(Y) \succeq 0,
\end{equation}
i.e $\Psi(Y^\dagger Y) - \Psi(Y)^\dagger \Psi(Y)$ is positive semi-definite.
\end{lemma}
Note that this inequality in~\cite{paulsen2002completely} is stated for 2-positive maps, and can thus also be applied to all completely positive maps. Now, we can prove the operator norm bound:

\begin{proof}[Proof of \Cref{lem:unital-channel-operator-norm-upper-bound}]

It will be helpful to work with the channel $\Phi$ instead of its Liouville representation $M_\Phi$. By the definition of the Liouville representation, we have
\begin{equation}
    \|M_\Phi\|_{\infty} = \max_{x \in \mathbb{C}^{d^2}, x \neq 0} \frac{\|M_\Phi x\|_2}{\|x\|_2} = \max_{X \in \mathbb{C}^{d \times d}, X \neq 0} \frac{\|M_\Phi \mathrm{vec}(X)\|_2}{\|X\|_2} = \max_{X \in \mathbb{C}^{d \times d}, X \neq 0} \frac{\|\Phi(X)\|_2}{\|X\|_2} = \|\Phi\|_{2 \rightarrow 2}.
\end{equation}

To bound this quantity, we will use the adjoint map $\Phi^*$, i.e., the channel satisfying
\begin{equation}
    \langle Y, \Phi(X) \rangle = \langle \Phi^*(Y), X\rangle, \quad \forall X \in \mathbb{C}^{d \times d}, Y \in \mathbb{C}^{d_q \times d_q}.
\end{equation}

Let us state a few helpful properties of $\Phi^*$. First, assume $\Phi$ has Kraus operators $\{A_k\}_k \subseteq \mathbb{C}^{d_q \times d}$. As $\Phi$ is trace preserving, we see that the Kraus decomposition is normalized, i.e., $\sum_k A_k^\dag A_k = \mathbbm{1}_d$. Now, by definition of the adjoint map $\Phi^*$, it has Kraus operators $A_k^\dag$, implying that $\Phi^*$ is also completely positive. Now, note that
\begin{equation}
    \Phi^*(\mathbbm{1}_{d_q}) = \sum_{k} A_k^\dag A_k = \mathbbm{1}_d,
\end{equation}
where we used the normalization of $\Phi$'s Kraus decomposition. Thus, $\Phi^*$ is unital; we will make use of this fact to apply \Cref{lem:kadison-schwarz} to $\Phi^*$ later.

Now, by definition of the adjoint map, we can show $\|\Phi^*\|_{2 \rightarrow 2} = \|\Phi\|_{2 \rightarrow 2}$ as follows:
\begin{align}
    \|\Phi\|_{2 \rightarrow 2} &= \sup_{X \in \mathbb{C}^{d \times d}, \|X\|_2 = 1} \|\Phi(X)\|_2
    \\&= \sup_{X \in \mathbb{C}^{d \times d}, \|X\|_2 = 1} \hspace{.5em} \sup_{Y \in \mathbb{C}^{d_q \times d_q}, \|Y\|_2 = 1} |\langle Y, \Phi(X)\rangle|
    \\&= \sup_{Y \in \mathbb{C}^{d_q \times d_q}, \|Y\|_2 = 1} \hspace{.5em} \sup_{X \in \mathbb{C}^{d \times d}, \|X\|_2 = 1} |\langle \Phi^*(Y), X\rangle|
    \\&= \sup_{Y \in \mathbb{C}^{d_q \times d_q}, \|Y\|_2 = 1} \|\Phi^*(Y)\|_2 = \|\Phi^*\|_{2 \rightarrow 2}.
\end{align}

Now, applying \Cref{lem:kadison-schwarz} to the unital and completely positive map $\Phi^*$, we get
\begin{equation}
    \Phi^*(Y^\dag Y) \succeq \Phi^*(Y)^\dag \Phi^*(Y) \implies \tr(\Phi^*(Y^\dag Y)) \geq \tr(\Phi^*(Y)^\dag \Phi^*(Y)) = \|\Phi^*(Y)\|_2^2, \label{eq:operator-norm-1}
\end{equation}
for all $Y \in \mathbb{C}^{d_q \times d_q}$,

Note that $\Phi^*$ is completely positive and thus also Hermiticity-preserving; as $Y^\dag Y$ is Hermitian, this implies $\Phi^*(Y^\dag Y)$ is Hermitian too. Now, we can rewrite $\tr(\Phi^*(Y^\dag Y))$ as follows:
\begin{align}
    \tr(\Phi^*(Y^\dag Y)) &= \langle \Phi^*(Y^\dag Y), \mathbbm{1}_d \rangle
    \\&= \langle Y^\dag Y, \Phi(\mathbbm{1}_d) \rangle
    \\&= \frac{d}{d_q} \langle Y^\dag Y, \mathbbm{1}_{d_q} \rangle
    \\&= \frac{d}{d_q} \tr(Y^\dag Y) = \frac{d}{d_q} \cdot \|Y\|_2^2, \label{eq:operator-norm-2}
\end{align}
where we used the fact that $\Phi^*(Y^\dag Y)$ is Hermitian in the first step and the assumption that $\Phi\left(\frac{\mathbbm{1}_d}{d}\right) = \frac{\mathbbm{1}_{d_q}}{d_q}$ in the third step. Now, using \Cref{eq:operator-norm-1,eq:operator-norm-2}, we have
\begin{align}
    \|\Phi^*\|_{2 \rightarrow 2} &= \sup_{Y \in \mathbb{C}^{d_q \times d_q}, \|Y\|_2 = 1} \|\Phi^*(Y)\|_2
    \\&\leq \sup_{Y \in \mathbb{C}^{d_q \times d_q}, \|Y\|_2 = 1} \sqrt{\tr(\Phi^*(Y^\dag Y))}
    \\&= \sup_{Y \in \mathbb{C}^{d_q \times d_q}, \|Y\|_2 = 1} \sqrt{\frac{d}{d_q}} \cdot \|Y\|_2 = \sqrt{\frac{d}{d_q}}.
\end{align}

Finally, we have shown
\begin{equation}
    M_{\Phi} = \|\Phi\|_{2 \rightarrow 2} = \|\Phi^*\|_{2 \rightarrow 2} \leq \sqrt{\frac{d}{d_q}},
\end{equation}
concluding the proof.
\end{proof}

\appendix

\section{Centralized mixedness testing lower bound with the \cite{liu2024quantum} instance}
\label{s:centralized-mixedness-testing-bound}
In this section, we use \Cref{lem:quantum-ingster-suslina} and the hard instance from \Cref{def:hard-instance-mic} to provide an alternate proof of the $\Omega(d/\eps^2)$ copy complexity lower bound for mixedness testing in the fully unrestricted setting.

\begin{theorem}
    Let $d \geq 6, 0 < \eps < C$. For the ensemble given in \Cref{def:hard-instance-mic},
    \begin{equation}
        \qdchi\left(\mathbb{E}_{\bfz}[\rho_{\bfz}^{\otimes n}] \big\| \left(\mathbbm{1}/d\right)^{\otimes n}\right) \leq \exp\left(\frac{n^2c^4\eps^4}{2\ell^2}\right) - 1 + \frac{4}{e^d}.
    \end{equation}
    Thus, the quantum $\chi^2$-divergence is $< \frac{1}{16}$ unless $n = \Omega(\sqrt{\ell}/\eps^2)$. Setting $\ell = d^2-1$ recovers the $\Omega(d/\eps^2)$ lower bound for mixedness testing.
\end{theorem}

\begin{proof}
    Keeping \Cref{lem:quantum-ingster-suslina} in mind, we wish to bound the moment generating function of $Z(\bfz,\bfz^\prime) \triangleq d\cdot\tr(\bar{\Delta}_{\bfz} \bar{\Delta}_{\bfz^\prime})$. Recall that $N_z \triangleq \min\{1, \frac{1}{d\|\Delta_z\|_\infty}\}$. Then,
    \begin{align}
        Z(z,z^\prime) &= d\cdot\tr(\bar{\Delta}_z \bar{\Delta}_{z^\prime}) 
        \\&= d \cdot \frac{c^2\eps^2}{d\ell} \sum_{i,j \in [\ell]} \tr(z_iz^\prime_j V_i V_j) N_z N_{z^\prime}
        \\&= \frac{c^2\eps^2}{\ell}N_zN_{z^\prime} \sum_{i \in [\ell]} z_i z_i^\prime 
        \\&= \frac{c^2\eps^2}{\ell}N_zN_{z^\prime} z^\top z^\prime.
    \end{align}
    Now, by \Cref{lem:quantum-ingster-suslina}, we have
    \begin{align}
       \qdchi\left(\mathbb{E}_{\bfz}[\rho_{\bfz}^{\otimes n}] \big\| \left(\mathbbm{1}/d\right)^{\otimes n}\right) + 1 &\leq \mathbb{E}_{\bfz,\bfz^\prime}\left[\exp(n Z(\bfz,\bfz^\prime))\right]
       \\&= \mathbb{E}_{\bfz,\bfz^\prime} \left[\exp\left(\frac{nc^2\eps^2}{\ell} N_{\bfz} N_{{\bfz}^\prime} \cdot \bfz^\top \bfz^\prime\right)\right]
       \\&\leq \mathbb{E}_{\bfz,\bfz^\prime} \left[\exp\left(\frac{nc^2\eps^2}{\ell} \cdot \bfz^\top \bfz^\prime\right)\right] + \frac{4}{e^d}
       \\&= \prod_{i = 1}^{\ell}\left(\mathbb{E}_{\bfz_i,\bfz_i^\prime} \left[
            \exp\left(\frac{nc^2\eps^2}{\ell} \cdot \bfz_i \bfz_i^\prime\right)
       \right]\right) + \frac{4}{e^d}
       \\&= \left(\mathbb{E}_{\boldsymbol{b} \sim \{-1,+1\}} \left[
            \exp\left(\frac{nc^2\eps^2}{\ell} \cdot \boldsymbol{b}\right)
       \right]\right)^{\ell} + \frac{4}{e^d}
       \\&\leq \left(\exp\left(\frac{n^2c^4\eps^4}{2\ell^2}\right)\right)^{\ell} + \frac{4}{e^d}
       \\&= \exp\left(\frac{n^2c^4\eps^4}{2\ell}\right) + \frac{4}{e^d},
    \end{align}
    where the second inequality follows from an argument identical to the proof of \cite[Lemma B.8]{liu2024quantum}, the next step uses the fact that the entries of $\bfz$ are drawn independently, and the subsequent step treats each product $\bfz_i\bfz_i^\prime$ as a uniformly random sign. In the last inequality, we used $\frac{e^x + e^{-x}}{2} \leq e^{x^2/2}$.
\end{proof}

\section{Proof of \Cref{lem:copy-complexity-and-quantum-chi2}}
\label{s:maxmin-minmax-divergence}

\begin{proof}
    Let us first consider the public-coin setting. As discussed in \Cref{ss:technical-overview}, we can treat the action of the distributed nodes as a set of channels $\{\Phi_i\}_{i \in [m]}$, where each $\Phi_i$ can depend on some shared random string $\boldsymbol{r} \in \{0,1\}^*$. Let $\mc{A}(\rho,\boldsymbol{r})$ denote the output of the overall algorithm implemented by the central and distributed nodes on input $\rho$ and shared random string $\boldsymbol{r}$. Let $m$ be large enough for the algorithm to succeed at certification with probability at least $\frac34$. Then, for any almost-$\eps$ perturbation $D \in \mc{D}_\eps\bigl(\frac{\mathbbm{1}}{d}\bigr)$, we have
    \begin{align}
        \frac12 \mathbf{Pr}_{\rho \sim D}[\mc{A}(\rho, \boldsymbol{r}) = \text{`Reject'}] + \frac12 \mathbf{Pr}_{\rho = \mathbbm{1}/d}[\mc{A}(\rho, \boldsymbol{r}) = \text{`Accept'}] &\geq \frac12 \cdot \frac{3}{4} \cdot \mathbf{Pr}_{\rho \sim D}\left[\left\|\rho - \frac{\mathbbm{1}}{d}\right\|_1 \geq \eps\right] + \frac12 \cdot \frac{3}{4} 
        \\&\geq \frac{9}{16}. 
    \end{align}
    Consequently, for any almost-$\eps$ perturbation $D$, there exists some instantiation $\boldsymbol{r} = r$ for which
    \begin{align}
        \frac{1}{2} \mathbf{Pr}_{\rho \sim D}[\mc{A}(\rho, r) = \text{`Reject'} | \boldsymbol{r} = r] + \frac12 \mathbf{Pr}_{\rho = \mathbbm{1}/d}[\mc{A}(\rho, r) = \text{`Accept'} | \boldsymbol{r} = r] \geq \frac{9}{16}.
    \end{align}
    In other words, for each ensemble $D$, there exists a deterministic choice of channels $\Phi_1, \dots, \Phi_m$ for which the central node can distinguish between the equally likely cases of  $\rho \sim D$ and $\rho = \mathbbm{1}/d$ with probability at least $\frac{9}{16}$. Consequently, by the Helstrom bound, we have
    \begin{equation}
        \frac12 + \frac14 \left\|\mathbb{E}_{\rho \sim D}\left[\bigotimes_{i = 1}^m \Phi_i(\rho)\right] - \bigotimes_{i = 1}^m \Phi_i\left(\mmstate\right) \right\|_1 \geq \frac{9}{16},
    \end{equation}
    or equivalently, by maximizing over $\Phi_1, \dots, \Phi_m$ and minimizing over $D$, we have
    \begin{equation}
        \min_{D \in \mc{D}_\eps(\mathbbm{1}/d)} \max_{\Phi_1, \dots, \Phi_m} \left\|
        \mathbb{E}_{\rho \sim D}\left[\bigotimes_{i = 1}^m \Phi_i(\rho)\right] - \bigotimes_{i = 1}^m \Phi_i\left(\mmstate\right)
        \right\|_1 \geq \frac14.
    \end{equation}
    Now, the first part of \Cref{lem:copy-complexity-and-quantum-chi2} follows by relating the trace distance above to the quantum $\chi^2$-divergence using \Cref{eq:dtr-qdchi}.

    In the private-coin setting, we can imagine that each node $N_i$ receives its own random bits $\boldsymbol{r}_i \in \{0,1\}^*$. Conditioned on $\boldsymbol{r}_i = r$, let its action be $\Phi_i^r$. Then, its overall action is given by $\mathbb{E}_{\boldsymbol{r}_i}[\Phi_i^{\boldsymbol{r}_i}]$; however, this is a mixture of quantum channels, and is thus a quantum channel itself. Consequently, we can model the action of each distributed node $N_i$ as a deterministic channel $\Phi_i$. Further, these deterministically chosen channels $\Phi_1, \dots, \Phi_m$ must allow the central node to succeed at $\eps$-certification, and thus, by the arguments above, we must have
    \begin{equation}
        \left\|
        \mathbb{E}_{\rho \sim D}\left[\bigotimes_{i = 1}^m \Phi_i(\rho)\right] - \bigotimes_{i = 1}^m \Phi_i\left(\mmstate\right)
        \right\|_1 \geq \frac14,
    \end{equation}
    for \emph{all} choices of $D$. Consequently, we must have 
    \begin{equation}
        \max_{\Phi_1, \dots, \Phi_m} \min_{D \in  \mc{D}_\eps(\mathbbm{1}/d)} \left\|
        \mathbb{E}_{\rho \sim D}\left[\bigotimes_{i = 1}^m \Phi_i(\rho)\right] - \bigotimes_{i = 1}^m \Phi_i\left(\mmstate\right)
        \right\|_1 \geq \frac14.
    \end{equation}
    The second part of the lemma can now be obtained using \Cref{eq:dtr-qdchi}.
\end{proof}

\section{Exponential separations from shared entanglement}\label{s:entanglement}

In this section, we demonstrate that for certain testing problems for $n$-qubit states, two of our distributed inference settings, namely the settings of  $(n_c > 0, n_q = 0, R = \mathsf{public}, E = 0)$ and $(n_c = 2n, n_q = 0, R = \mathsf{private}, E = n)$, are exponentially separated. First, note that for unbounded $n_c$, the former setting is equivalent to the centralized setting where testers can perform non-adaptive single-copy measurements. In this centralized setting, problems such as purity testing and unsigned Pauli shadow tomography are known to require exponentially many copies \cite{chen2022exponential}. 

However, given access to just $2$-copy measurements, one can solve the above problems with constant complexity by performing \emph{Bell sampling}. The key insight that leads to our exponential separation is the observation that one can also perform Bell sampling in the $(n_c = 2n, n_q = 0, R = \mathsf{private}, E = n)$ setting.

Let us first define Bell sampling. For any string $x = (a,b) \in \{0,1\}^{2n}$, we define the Weyl operators
\begin{equation}
    W_x \triangleq i^{a.b} (X^{a_1}Z^{b_1}) \otimes \dots \otimes (X^{a_n}Z^{b_n}),
\end{equation}
where $X,Z$ are the single-qubit Pauli matrices. It is not hard to check that the Weyl operators are self-adjoint, i.e. $W_x = W_x^\dagger$.
We now define the Bell state associated with each $x \in \{0,1\}^{2n}$.
\begin{equation}
    \ket{\psi_x} \triangleq (W_x \otimes I) \ket{\epr_n} = \frac{i^{a.b}}{\sqrt{2^n}} \sum_{k \in \{0,1\}^n} (-1)^{k.b} \ket{k\oplus a, k}.
\end{equation}
We note that the $4^n$ Bell states form an orthonormal basis for $\mathbb{C}^{4^n}$, called the \emph{Bell basis}. Measuring two copies of a quantum state $\rho$ in the Bell basis results in a \emph{Bell sample} $x \in \{0,1\}^{2n}$, with probability 
\begin{equation}
    p_\rho(x) \defeq \bra{\psi_x} \rho \otimes \rho \ket{\psi_x}.
\end{equation}

Bell sampling is a standard procedure in quantum learning and testing, and indeed leads to efficient algorithms for unsigned Pauli shadow tomography and purity testing. Let us now formally state our main observation:

\begin{proposition}
\label{obs:distributed-bell-sampling}
    Given two distributed nodes $N_1, N_2$ in the $(n_c = 2n, n_q = 0, R = \mathsf{private}, E = n)$ setting, where each distributed node receives a copy of an $n$-qubit mixed state $\rho$, there exists an algorithm allowing the central node $N_c$ to exactly sample from the distribution $p_\rho$ with probability $1$.
\end{proposition}

Before proving the above proposition, note that this immediately implies efficient algorithms for purity testing and unsigned shadow tomography in our distributed setting, by pairing up adjacent distributed nodes to generate multiple Bell samples. Consequently, we obtain exponential separations between the two settings mentioned at the start of this section. For simplicity, we only state such a separation formally for purity testing, but note that it easily extends to any task which can be solved efficiently via Bell sampling and is hard in the setting of non-adaptive single-copy measurements. By the observation above and \cite[Theorem 5.11]{chen2022exponential}, we have the following separation:

\begin{corollary}
    For any $n$-qubit quantum state $\rho$, in the $(n_c > 0, n_q = 0, R = \mathsf{public}, E = 0)$-setting, the number of distributed nodes needed to test whether $\rho$ is a pure state or the maximally mixed state is at least $\Omega(2^{n/2})$. However, in the $(n_c = 2n, n_q = 0, R = \mathsf{private}, E = n)$-setting, only $\bigo(1)$ distributed nodes suffice to perform such a test.
\end{corollary}

It remains now to prove \Cref{obs:distributed-bell-sampling}.

\begin{proof}[Proof of \Cref{obs:distributed-bell-sampling}]
    Our distributed Bell sampling routine is straightforward and inspired by quantum teleportation. However, despite sharing sufficiently many EPR pairs, the distributed nodes cannot quite perform teleportation; as we do not allow any communication (even classical) between $N_1$ and $N_2$, they cannot communicate the post-measurement corrections necessary for teleportation. Instead, they will send these corrections to the central node $N_c$.

    Let us assume that each node places its copy of $\rho$ in register $A_i$ and its half of the EPR pairs in register $B_i$. Then, the joint state of the two nodes is given by:
    \begin{equation}
        \sigma = \rho_{A_1} \otimes \rho_{A_2} \otimes \ket{\epr_n}\bra{\epr_n}_{B_1,B_2},
    \end{equation}
    where $\ket{\epr_n}$ denotes $n$ copies of the maximally entangled state.

    Now, for $i \in \{1,2\}$, each node $N_i$ measures its registers $A_i,B_i$ in the Bell basis, and communicates its output string $z_i \in \{0,1\}^{2n}$ to $N_c$. $N_c$ will then simply output $z = z_1 \oplus z_2$, and we will show that this is distributed according to $p_\rho(z)$.

    Without loss of generality, we can assume $N_1$ performs its Bell basis measurements before $N_2$. When discussing $N_1$'s measurements, we will only be concerned with the registers $A_1, B_1, B_2$. The measurement performed by $N_1$ is then described by the POVM consisting of orthogonal projectors $\{\Pi_x = \ketbra{\psi_x}{\psi_x}_{A_1,B_1} \otimes I_{B_2} \}_{x \in \{0,1\}^{2n}}$. Let the outcome of this measurement on $\sigma_{A_1,B_1,B_2}$ be some string $z_1 = (a,b) \in \{0,1\}^{2n}$.

    Let us compute the post-measurement state conditioned on obtaining outcome $z_1$. For ease of exposition, we first assume $\rho$ is a pure state, i.e., we write $\rho = \ket{\phi}\bra{\phi}$ for some $\ket{\phi} \triangleq \sum_{i \in \{0,1\}^n} \alpha_i \ket{i} \in \mathbb{C}^{2^n}$ and amplitudes $\alpha_i \in \mathbb{C}$.
    Now, to compute the post-measurement states, let us compute $\Pi_{z_1} (\ket{\phi}_{A_1} \otimes \ket{\epr_n}_{B_1,B_2})$.
    \begin{align}
         \Pi_{z_1} (\ket{\phi}_{A_1} \otimes \ket{\epr_n}_{B_1,B_2}) &= (\ketbra{\psi_{z_1}}{\psi_{z_1}}_{A_1,B_1} \otimes I_{B_2}) (\ket{\phi}_{A_1} \otimes \ket{\epr_n}_{B_1,B_2})
         \\&= \frac{(-i)^{a.b}}{2^n} \sum_j (-1)^{j.b} (\ketbra{\psi_{z_1}}{j \oplus a,j}_{A_1,B_1} \otimes I_{B_2}) \sum_{k, l} \alpha_k \ket{k,l,l}_{A_1,B_1,B_2}
         \\&= \frac{(-i)^{a.b}}{2^n} \ket{\psi_{z_1}}_{A_1,B_1} \sum_{j,k,l} (-1)^{j.b} \alpha_k \delta_{j \oplus a, k} \delta_{j,l} \ket{l}_{B_2}
         \\&= \frac{1}{2^n} \ket{\psi_{z_1}}_{A_1,B_1} (-i)^{a.b} \sum_{k} (-1)^{k.b} (-1)^{a.b} \alpha_{k} \ket{k \oplus a}_{B_2}
         \\&= \frac{1}{2^n} \ket{\psi_{z_1}}_{A_1,B_1} \otimes W_{z_1} \ket{\phi}_{B_2}.
     \end{align}
     Thus, for a generic mixed state, by linearity, we can write
     \begin{equation}
         \Pi_{z_1} (\rho \otimes \ket{\epr_n}\bra{\epr_n}_{B_1,B_2}) \Pi_{z_1} = \frac{1}{4^n} \ket{\psi_{z_1}}\bra{\psi_{z_1}}_{A_1,B_1} \otimes W_{z_1} \rho W_{z_1}^\dag.
     \end{equation}
     Now, we clearly have $\tr(\Pi_{z_1} (\rho \otimes \ket{\epr_n}\bra{\epr_n}_{B_1,B_2})) = \frac{1}{4^n}$, i.e., each outcome $z_1$ occurs with probability $\frac{1}{4^n}$, and the post-measurement state conditioned on measuring $z_1$ is $\ket{\psi_{z_1}}\bra{\psi_{z_1}}_{A_1,B_1} \otimes W_{z_1} \rho W_{z_1}^\dag$.

    Thus, conditioned on $z_1$, $N_2$'s mixed state is given by $W_{z_1} \rho W_{z_1}^\dagger \otimes \rho$. Measuring this state in the Bell basis, $N_2$ obtains outcome $z_2$ with probability
    \begin{align}
        \Pr(z_2 | z_1) &= \bra{\psi_{z_2}} W_{z_1} \rho  W_{z_1}^\dagger \otimes \rho \ket{\psi_{z_2}}
        \\&=  \bra{\psi_{z_1 \oplus z_2}} \rho \otimes \rho \ket{\psi_{z_1 \oplus z_2}}
        \\&= p_\rho(z_1 \oplus z_2).
    \end{align}
    Finally, the probability of $N_c$ outputting a string $z$ is given by
    \begin{align}
        \Pr(z) &= \sum_{z_1,z_2 : z_1 \oplus z_2 = z}  \Pr(z_1) \Pr(z_2| z_1) 
        \\&= \sum_{z_1,z_2 : z_1 \oplus z_2 = z} \frac{1}{4^n} p_\rho(z_1 \oplus z_2)
        \\&= p_\rho(z).
    \end{align}
    Thus, the described algorithm produces a single sample from the distribution $p_\rho$, as claimed.     
\end{proof}

\printbibliography

\end{document}

%% file: usepackages.tex
\usepackage[utf8]{inputenc}
\usepackage{lipsum}

\usepackage{amsmath, amssymb,amsfonts,amsthm,mathtools}
\usepackage{xspace,graphicx,relsize,bm,bbm,xcolor}
\usepackage{soul} 

\usepackage{parskip}  

\usepackage{libertine}
\usepackage{libertinust1math}
\usepackage{dsfont}
\usepackage[T1]{fontenc}
\usepackage{enumitem}
\usepackage{nicefrac}

\usepackage[
    backend=biber,
    style=alphabetic,
    sorting=anyt,
    minalphanames=3,
    maxalphanames=3,
    maxnames=99,
    backref=true
    ]{biblatex}
\DefineBibliographyStrings{english}{%
  backrefpage = {page},
  backrefpages = {pages},
}

\usepackage{sepfootnotes}
\newendnotes{x}
\renewcommand\xnotesize\normalsize

\usepackage{hyperref}
\usepackage[margin=1.75cm]{geometry}
\definecolor{linkcol}{rgb}{0.0,0.55,0.7}
\definecolor{citecol}{rgb}{0.0, 0.6, 0.45}
\definecolor{urlcol}{rgb}{0.7, 0.0, 0.55}
\hypersetup{
	colorlinks,
	linkcolor={linkcol},
	citecolor={citecol},
	urlcolor={urlcol}
}

\usepackage{url}
\usepackage{subcaption}
\usepackage{mleftright}
\usepackage{hyperref}
\usepackage{multirow}
\usepackage{physics}  

\usepackage{algorithm}
\usepackage{algpseudocodex}[indLines = true,italicComments = false]

\usepackage{zref-clever}
\zcsetup{cap = true}
\newcommand{\Cref}{\zcref}

\usepackage{authblk}  

%% file: commands.tex
\def\01{\{0,1\}}
\newcommand{\mc}[1]{\mathcal{#1}}

\newcommand{\defeq}{:=}

\let\Pr\relax
\DeclareMathOperator*{\Pr}{\mathbf{Pr}}

\newcommand{\eps}{\epsilon}

\newcommand{\swap}{\mathtt{SWAP}}

\newcommand{\epr}{\mathtt{EPR}}

\newcommand{\bigo}{\mathcal{O}}
\newcommand{\wg}{\mathrm{Wg}}
\newcommand{\qdchi}{\mathrm{D}_{\chi^2}}
\newcommand{\bftheta}{\boldsymbol{\theta}}
\newcommand{\mmstate}{\frac{\mathbbm{1}}{d}}
\newcommand{\bfz}{\boldsymbol{z}}
\newcommand{\bfU}{\boldsymbol{U}}

%% file: mytheorems.tex
\newtheoremstyle{mydefinitionsty}
{10pt}
{10pt}
{}
{}
{}
{}
{.5em}
{\textbf{\thmname{#1}~\thmnumber{#2}:  }\thmnote{(#3)}}

\newtheoremstyle{myproblemsty}
{10pt}
{10pt}
{}
{}
{}
{}
{.5em}
{\textbf{\thmname{#1}~\thmnumber{#2}:  }\thmnote{(#3)}\newline}

\newtheoremstyle{mythmsty}
{10pt}
{10pt}
{\itshape}
{}
{}
{}
{.5em}
{\textbf{\thmname{#1}~\thmnumber{#2}:  }\thmnote{(#3)}}
\theoremstyle{mythmsty}
\newtheorem{theorem}{Theorem}[section]
\newtheorem{proposition}[theorem]{Proposition}
\newtheorem{lemma}[theorem]{Lemma}
\newtheorem{corollary}[theorem]{Corollary}

\AddToHook{env/theorem/begin}{%
\zcsetup{countertype={theorem=theorem}}}
\AddToHook{env/proposition/begin}{%
\zcsetup{countertype={theorem=proposition}}}
\AddToHook{env/lemma/begin}{%
\zcsetup{countertype={theorem=lemma}}}
\AddToHook{env/corollary/begin}{%
\zcsetup{countertype={theorem=corollary}}}
\AddToHook{env/claim/begin}{%
\zcsetup{countertype={theorem=claim}}}
\AddToHook{env/conjecture/begin}{%
\zcsetup{countertype={theorem=conjecture}}}
\AddToHook{env/fact/begin}{%
\zcsetup{countertype={theorem=fact}}}
\zcRefTypeSetup{fact}{
Name-sg = Fact ,
name-sg = fact ,
Name-pl = Facts ,
name-pl = facts ,
}
\AddToHook{env/question/begin}{%
\zcsetup{countertype={theorem=question}}}

\theoremstyle{mydefinitionsty}
\newtheorem{definition}[theorem]{Definition}

\newtheorem{remark}[theorem]{Remark}
\AddToHook{env/definition/begin}{%
\zcsetup{countertype={theorem=definition}}}
\AddToHook{env/notation/begin}{%
\zcsetup{countertype={theorem=notation}}}
\AddToHook{env/example/begin}{%
\zcsetup{countertype={theorem=example}}}
\AddToHook{env/assumption/begin}{%
\zcsetup{countertype={theorem=assumption}}}
\AddToHook{env/remark/begin}{%
\zcsetup{countertype={theorem=remark}}}

\theoremstyle{myproblemsty}

\AddToHook{env/problem/begin}{%
\zcsetup{countertype={theorem=problem}}}

\numberwithin{equation}{section}

%% file: comments.tex
\definecolor{alexcolor}{rgb}{0.0, 0.47, 0.75}   

\definecolor{questioncolor}{rgb}{0.36, 0.54, 0.66}